\documentclass[11pt]{article}

\usepackage{xspace,amstext,amssymb,url}
\usepackage{amsmath}
\usepackage{amsthm}
\usepackage[noend]{algorithmic}
\usepackage{graphicx,color,epsfig}
\usepackage{pstricks,pst-node}
\usepackage{subfigure}
\usepackage{enumerate}
 \usepackage{pgf}
 \usepackage{tikz}
 \usetikzlibrary{automata}
 \usetikzlibrary{arrows}
 \usepackage{pstricks}
\usepackage{helvet,url}
\usepackage{epsf}
\usepackage{graphicx}
\usepackage{ifthen,shortvrb}

\newtheorem{mytheo}{Theorem}

{\bfseries}{\itshape}
{\bfseries}{\itshape}
{\bfseries}{\itshape}
{\bfseries}{\itshape}
{\itshape}{\itshape}
\newtheorem{mylemma}[mytheo]{Lemma}{\bfseries}{\itshape}
{\itshape}{\itshape}

\newenvironment{myproof}{\begin{proof}}{\end{proof}}

\newtheorem{thm}{Theorem}[section]
\newtheorem{theorem}[thm]{Theorem}

\newtheorem{corollary}[thm]{Corollary}

\newtheorem{lemma}[thm]{Lemma}

\newtheorem{proposition}[thm]{Proposition}

\theoremstyle{definition}

\newtheorem{remark}[thm]{Remark}

\newtheorem{example}[thm]{Example}

\newtheorem{definition}[thm]{Definition}

\numberwithin{equation}{section}

\def\pushright#1{{
   \parfillskip=0pt            
   \widowpenalty=10000         
   \displaywidowpenalty=10000  
   \finalhyphendemerits=0      
   \leavevmode                 
   \unskip                     
   \nobreak                    
   \hfil                       
   \penalty50                  
   \hskip.2em                  
   \null                       
   \hfill                      
   {#1}                        
   \par}}                      

\def\qEd{\rule{1ex}{1ex}}
\def\qed{\pushright{\qEd}
    \penalty-700 \par\addvspace{\medskipamount}}

\newinsert\copyins
\skip\copyins=3pc
\count\copyins=1000 
\dimen\copyins=.5\textheight 

\newcommand{\nc}{\newcommand}

\newcommand{\dynamic}[1]{\ensuremath{\mbox{Dyn}(#1)}\xspace}
\newcommand{\dynamicAlt}[1]{\ensuremath{\mbox{Dyn-alt}(#1)}\xspace}

\nc{\acclang}[1]{\mbox{ACC(#1)}\xspace}
\nc{\ops}{\mbox{$\Delta$}\xspace}
\nc{\op}{\mbox{op}\xspace}

\nc{\struc}{\text{STRUC}\xspace}
\nc{\arity}{\text{arity}\xspace}

\nc{\accept}{\ensuremath{\text{ACC}}\xspace}
\nc{\Aext}{\cA^{\text{Ext}}}
\nc{\Wext}{\mbox{$W^{\text{Ext}}$}}

\nc{\emptyws}{\mbox{$E_n$}\xspace}
\nc{\emptyas}{\ensuremath{E^{\text{Ext}}_n}\xspace}

\nc{\word}{\text{word}}
\nc{\reset}{\text{reset}}
\nc{\ins}{\text{ins}}
\nc{\del}{\text{del}}
\nc{\equal}{\text{EQUAL}\xspace}
\nc{\equaln}[1]{\mbox{$\equal_{#1}$}\xspace}
\nc{\midd}{\text{MIDDLE}\xspace}
\nc{\mincons}{\text{min}}
\nc{\maxcons}{\text{max}}

\nc{\ite}{\mbox{\textbf{ite}}}
\newcommand{\type}[2]{\ensuremath{\langle #1,#2 \rangle}}

\nc{\dynamicclass}[1]{\ensuremath{\text{Dyn#1}}\xspace}
\nc{\dynC}{\dynamicclass{C}\xspace}
\nc{\dynprop}{\text{DynPROP}\xspace}
\nc{\dynqf}{\text{DynQF}\xspace}
\nc{\dynQF}{\text{DynQF}\xspace}
\nc{\EFO}{\text{EFO}\xspace}
\nc{\dynFO}{\text{DynFO}\xspace}
\nc{\dyntc}{\ensuremath{\text{DynTC}^0}\xspace}
\nc{\dynp}{\text{DynP}\xspace}

\nc{\FO}{\text{FO}\xspace}
\nc{\qf}{\text{QF-FO}\xspace}

\nc{\dynL}{\text{DYN-L}\xspace}
\nc{\dynstandL}{\text{DYNSTAND-L}\xspace}

\nc{\languageclass}[1]{\mbox{$#1_{\text{Lang}}$}\xspace}
\nc{\dynFOlang}{\languageclass{\dynFO}}
\nc{\dynproplang}{\languageclass{\dynprop}}
\nc{\dynqflang}{\languageclass{\dynqf}}

\nc{\SetSucc}{\mbox{$\mathcal{SUCC}$\xspace}}

\nc{\bifunctions}[2]{\mbox{$\mbox{#1} + \mbox{#2}$}\xspace}
\nc{\bifset}[2]{\bifunctions{#1}{#2}\xspace}
\nc{\bifexpl}[2]{\bifunctions{#1}{$\{\text{#2}\}$}\xspace}
\nc{\bif}[1]{\bifunctions{#1}{\mbox{BIF}}\xspace}
\nc{\precomp}[1]{\mbox{$#1 + \mbox{PRC}$}\xspace}

\nc{\succM}{\text{\scshape succ}\xspace}
\nc{\preM}{\text{\scshape pre}\xspace}
\nc{\minM}{\text{\scshape min}\xspace}
\nc{\maxM}{\text{\scshape max}\xspace}

\nc{\repr}[1]{\ensuremath{\langle #1\rangle}}

\nc{\plusfun}{\text{\scshape plus}\xspace}
\nc{\minfun}{\text{\scshape minus}\xspace}

\nc{\Op}{\ensuremath{\text{Op}}}
\nc{\Cl}{\ensuremath{\text{Cl}}}

\nc{\fr}{\ensuremath{f^\rightarrow}}
\nc{\fl}{\ensuremath{f^\leftarrow}}
\nc{\gr}{\ensuremath{g^\rightarrow}}
\nc{\gl}{\ensuremath{g^\leftarrow}}

\nc{\mydef}{\ensuremath{=_{\textrm{def}}}}

\nc{\voc}{\ensuremath{\gamma}}

\nc{\lchild}{\text{L-child}\xspace}
\nc{\rchild}{\text{R-child}\xspace}
\nc{\lchildfun}{\text{l-child}\xspace}
\nc{\rchildfun}{\text{r-child}\xspace}
\nc{\anc}{\text{Anc}}
\nc{\ancself}{\text{Anc-self}}
\nc{\rootnode}{\text{root}}
\nc{\leaf}{\text{Leaf}}
\nc{\lca}{\text{lca}}
\nc{\nodes}{\text{nodes}}
\nc{\eps}{\text{Epsilon}}
\nc{\con}{\text{Con}}
\nc{\parent}{\text{parent}}

\nc{\ind}{\text{ind}\xspace}

\nc{\altreach}{\text{ALT-REACH}\xspace}
\nc{\reachable}{\mbox{Reach}\xspace}
\nc{\bdaltreach}[1]{\mbox{$\altreach_{\text{depth}\leq #1}$}\xspace}
\nc{\neighborhood}[1]{\mbox{$\mathcal{N}_{#1}$}\xspace}

\nc{\actdom}{\text{Act}\xspace}

\nc{\into}[2]{\mbox{${]#1, #2[}$}}
\nc{\intc}[2]{\mbox{${[#1, #2]}$}}

\nc{\shuf}{\&}
\nc{\shb}{\mathbin\&}
\nc{\shuffle}{\shuf}
\nc{\mycount}{\#}
\nc{\re}{\text{RE}\xspace}
\nc{\rein}{\text{RE($\cap$)}\xspace}
\nc{\recom}{\text{RE($\neg$)}\xspace}
\nc{\reincom}{\text{RE($\cap$,$\neg$)}\xspace}
\nc{\res}{\text{REs}\xspace}
\nc{\rec}{\text{RE($\mycount$)}\xspace}
\nc{\recs}{\text{RE($\mycount$)}s\xspace}
\nc{\rei}{\text{RE($\shuffle$)}\xspace}
\nc{\reis}{\text{RE($\shuffle$)}s\xspace}
\nc{\reci}{\text{RE($\mycount,\shuf$)}\xspace}
\nc{\recis}{\text{RE($\mycount,\shuf$)s}\xspace}
\nc{\chare}{\text{CHARE}\xspace}
\nc{\charec}{\text{CHARE($\mycount$)}\xspace}
\nc{\lin}{\text{RE(linear)}\xspace}
\nc{\strongLin}{\text{RE(strong-linear)}\xspace}
\nc{\disjoint}{\ensuremath \text{$^{\text{disj}}$}\xspace}
\nc{\strongLinDis}{\ensuremath \text{\strongLin$^{\text{disj}}$} \xspace}

\nc{\sym}{\text{Sym}\xspace}

\nc{\configuration}{\text{configuration}\xspace}
\nc{\configurations}{\text{configurations}\xspace}

\nc{\nta}{\text{NTA}\xspace}
\nc{\nfa}{\text{NFA}\xspace}
\nc{\nfas}{\text{NFAs}\xspace}
\nc{\nfaci}{\text{NFA$(\mycount,\shuffle)$}\xspace}
\nc{\nfacis}{\text{NFA$(\mycount,\shuffle)$s}\xspace}
\nc{\nfac}{\text{NFA$(\mycount)$}\xspace}
\nc{\nfacs}{\text{NFA$(\mycount)$s}\xspace}
\nc{\nfai}{\text{NFA$(\shuffle)$}\xspace}
\nc{\nfais}{\text{NFA$(\shuffle)$s}\xspace}

\nc{\patree}{\ensuremath \text{$(\re,R)_{\exists}$}\xspace}
\nc{\patreu}{\ensuremath \text{$(\re,R)_{\forall}$}\xspace}
\nc{\patlineaire}{\ensuremath \text{$(lineair,R)_{\exists}$}\xspace}
\nc{\patlineairu}{\ensuremath \text{$(lineair,R)_{\forall}$}\xspace}
\nc{\patstronge}{\ensuremath \text{$(stronglineair(\cap = \emptyset),R)_{\exists}$}\xspace}
\nc{\patstrongu}{\ensuremath \text{$(stronglineair(\cap = \emptyset),R)_{\forall}$}\xspace}

\nc{\nat}{\Bbb{N}}
\nc{\Dom}{\text{Dom}}
\nc{\lab}{\text{lab}}
\nc{\depth}{\text{depth}}

\nc{\ac}{\ensuremath{\text{AC}^0}\xspace} 
\nc{\logspace}{{\sc logspace}\xspace} 
\nc{\nlogspace}{{\sc nlogspace}\xspace} 
\nc{\logcfl}{\text{LOGCFL}\xspace} 
\nc{\ptime}{\text{PTIME}\xspace} 
\nc{\np}{{\sc np}\xspace} 
\nc{\conp} {\textrm{co}\textsc{np}\xspace}
\nc{\pspace}{{\sc pspace}\xspace} 
\nc{\exptime}{{\sc exptime}\xspace}
\nc{\expspace}{{\sc expspace}\xspace}
\nc{\nexpspace}{{\sc nexpspace}\xspace}
\nc{\nexptime}{{\sc nexptime}\xspace}
\nc{\blogspace}{{\scriptsize {\bf LOGSPACE}}}
\nc{\bnlogspace}{{\scriptsize {\bf  NLOGSPACE}}}
\nc{\bptime}{{\scriptsize {\bf PTIME}}}
\nc{\bnp}{{\scriptsize {\bf NP}}}
\nc{\bconp}{{\scriptsize {\bf coNP}}}
\nc{\bpspace}{{\scriptsize {\bf PSPACE}}}
\nc{\bexptime}{{\scriptsize {\bf EXPTIME}}}
\nc{\bnexptime}{{\scriptsize {\bf NEXPTIME}}}
\nc{\bexpspace}{{\scriptsize {\bf EXPSPACE}}}

\nc{\expo}{\text{exp}\xspace}

\nc{\cO}{{\mathcal O}}
\nc{\cB}{{\mathcal B}}
\nc{\mcC}{{\mathcal C}}
\nc{\cR}{{\mathcal R}}
\nc{\cS}{{\mathcal S}}
\nc{\cF}{{\mathcal F}}
\nc{\cT}{{\mathcal T}}
\nc{\cM}{{\mathcal M}}
\nc{\cA}{{\mathcal A}}
\nc{\subs}{\subseteq}
\renewcommand{\epsilon}{\varepsilon}

\nc{\first}{\text{first}\xspace}
\nc{\last}{\text{last}\xspace}
\nc{\follow}{\text{Follow}\xspace}
\nc{\nfirst}{\text{Not-First}\xspace}
\nc{\nfollow}{\text{Not-Follow}\xspace}

\nc{\uab}{\text{RE(1UA)}}
\nc{\uabc}{\text{RE($\mycount$)-UAB}}
\nc{\uabw}{\ensuremath \text{UAB}_W}
\nc{\uabwc}{\ensuremath \text{UAB}^{\#}_W\xspace}
\nc{\uabs}{\ensuremath \text{UAB}_S}
\nc{\uabsc}{\ensuremath \text{UAB}^{\#}_S\xspace}

\nc{\satisfiability}{{\sc satisfiability}\xspace}
\nc{\inclusion}{{\sc inclusion}\xspace}
\nc{\intersection}{{\sc intersection}\xspace}
\nc{\equivalence}{{\sc equivalence}\xspace}
\nc{\emptiness}{{\sc emptiness}\xspace}
\nc{\membership}{{\sc membership}\xspace}
\nc{\simplification}{{\sc simplification}\xspace}
\nc{\reachability}{{\sc reachability}\xspace}
\nc{\onerun}{{\sc one run}\xspace}
\nc{\partition}{{\sc partition}\xspace}
\nc{\universality}{{\sc universality}\xspace}

\nc{\edtd}{\text{EDTD}\xspace}
\nc{\edtds}{\text{EDTDs}\xspace}
\nc{\edtdre}{\text{\edtd}\xspace}
\nc{\edtdc}{\text{\edtd}($\mycount$)\xspace}
\nc{\edtdi}{\text{\edtd}($\shuffle$)\xspace}
\nc{\edtdci}{\text{\edtd}($\mycount$,$\shuffle$)\xspace}
\nc{\dtd}{\text{DTD}\xspace}
\nc{\dtdc}{\text{\dtd}($\mycount$)\xspace}
\nc{\dtdi}{\text{\dtd}($\shuffle$)\xspace}
\nc{\dtdci}{\text{\dtd}($\mycount$,$\shuffle$)\xspace}

\nc{\suf}{\text{Suffix}\xspace}
\nc{\prefix}{\text{Prefix}\xspace}
\nc{\enc}{\text{enc}\xspace}
\nc{\construct}{\text{construct}\xspace}
\nc{\init}{\text{init}\xspace}

\newcommand{\subtree}{\text{subtree}\xspace}

\nc{\Rel}{\ensuremath{\text{Rel}}\xspace}
\nc{\Fun}{\ensuremath{\text{Fun}}\xspace}
\nc{\BIF}{\ensuremath{\text{BIF}}\xspace}

\newenvironment{myalg}[2]
{\begin{minipage}{#1}\centerline{#2}\begin{algorithmic}[1]}
{\end{algorithmic}\end{minipage}}

\nc{\mybput}[4]{\rput(#1,0.6){\rnode{#2}{$#3$}}\rput(#1,0){$#4$}}
\nc{\myline}[3]{\ncline[linestyle=#3]{#1}{#2}}

\newcommand{\ignore}[1]{}

\begin{document}

\title{The Dynamic  Complexity of Formal Languages\footnote{This paper contains the material presented at the 26th International
  Symposium on Theoretical Aspects of Computer Science (STACS 2009) \cite{GeladeMS09},
  extended with proofs and examples which were omitted there due to
  lack of space.\smallskip}
}
\author{Wouter Gelade\thanks{Hasselt University and Transnational University of Limburg, School for Information Technology. Wouter Gelade is a Research Assistant of the Fund for
  Scientific Research - Flanders (Belgium).} \and Marcel Marquardt\thanks{Technische Universit\"at Dortmund} \and Thomas Schwentick${}^{\ddag}$}

\maketitle

\begin{abstract}
  \noindent The paper investigates the power of the dynamic complexity
  classes \dynFO, \dynqf and \dynprop over string languages. The
  latter two classes contain problems that can be maintained using
  quantifier-free first-order updates, with and without auxiliary
  functions, respectively. It is shown that the languages maintainable in
  \dynprop exactly are the regular languages, even when
  allowing arbitrary precomputation.  This enables lower
  bounds for \dynprop and separates \dynprop from \dynqf and
  \dynFO. Further, it is shown that any context-free language can be
  maintained in \dynFO and a number of specific context-free
  languages, for example all Dyck-languages, are maintainable in
  \dynqf. Furthermore, the dynamic complexity of regular
  tree languages is investigated and some results concerning
  arbitrary structures are obtained: there exist first-order definable properties which are not maintainable in \dynprop. On the other hand any
  existential first-order property can be maintained in \dynqf when
  allowing precomputation.
\end{abstract}



\vspace{-.5cm}
\section{Introduction}

Traditional complexity theory asks for the necessary effort to decide
whether a given input has a certain property, more precisely, whether
a given string is in a certain language. In contrast, {\em dynamic
  complexity} asks for the effort to maintain sufficient knowledge to
be able to decide whether the input object has the property {\em after
  a series of small changes of the object}. 
 The complexity theoretic investigation of the dynamic complexity of
algorithmic problems was initiated by Patnaik and Immerman
\cite{PatnaikI97}. They defined the class \dynFO of dynamic problems
where small changes in the input can be mastered by formulas of
(first-order) predicate logic (or, equivalently, poly-size circuits of
bounded depth, see \cite{Etessami98}). More precisely, the dynamic program makes use of an
auxiliary data structure and after each update (say, insertion or
deletion) the auxiliary data structure can be adapted by a first-order
formula.

Among others they showed that the dynamic complexity of the following
problems is in \dynFO: Reachability in undirected graphs,
minimum spanning forests, multiplication, regular languages, the Dyck
languages $D_n$. Subsequent work has yielded more problems in \dynFO \cite{Etessami98}
some of which are \logcfl-complete \cite{WeberS05} and even
\ptime-complete \cite{MiltersenSVT94,PatnaikI97} (even though the
latter are highly artificial).  Other work also considered stronger
classes (like Hesse's result that Reachability in arbitrary directed
graphs is in \dyntc \cite{Hesse03}), studied notions of completeness
for dynamic problems \cite{HesseI02}, and elaborated on the handling
of precomputations \cite{WeberS05}.

The choice of first-order logic as update language in
\cite{PatnaikI97} was presumably triggered by the hope that, in the
light of lower bounds for $\ac$, it would be
possible to prove that certain problems do {\em not} have \dynFO
dynamic complexity.  As it is easy to show that every \dynFO problem is in \ptime, a
non-trivial lower bound result would have to show that the dynamic complexity
of some \ptime problem is not in \dynFO. However, so far there are
no results of this kind. 

The inability to prove lower bounds has naturally led to the
consideration of subclasses of  \dynFO. Hesse studied
 problems with quantifier-free update formulas, yielding \dynprop if
 the maintained data structure is purely relational and \dynQF if
 functions are allowed as well \cite{Hesse03b,Hesse03a}. As further 
refinements the subclasses DynOR and DynProjections were studied. In
\cite{Hesse03b} separation results for subclasses of \dynprop were
  shown and the separation between \dynprop and \dynp was stated as an
  open problem.

The framework of \cite{PatnaikI97} allows more
general update operations and some of the results we mention depend on
the actual choice of operations. Nevertheless, most research has concentrated on
insertions and deletions as the only available
operations. Furthermore, most work considered underlying structures of
the following three kinds.
\begin{description}
\item[Graphs] Here, edges can be inserted or deleted. One of the main
  open questions is whether Reachability (aka transitive closure) can
  be maintained in \dynFO for directed, possibly cyclic graphs.
\item[Strings] Here, letters can be inserted or deleted. As mentioned
  above, \cite{PatnaikI97} showed that regular languages and Dyck
  languages can be maintained in \dynFO. Later on, Hesse proved that
  the dynamic complexity of regular languages is actually in \dynQF \cite{Hesse03a}.
\item[Databases] The dynamic complexity of database properties were
  studied in the slightly different framework of First-Order
  Incremental Evaluation Systems (FOIES) \cite{DongST95}. Many interesting results
  were shown including a separation between deterministic and
  nondeterministic systems~\cite{DongS97} and inexpressibility results for auxiliary
  relations of small arity ~\cite{DongLW93,DongS95c}.
Nevertheless, general lower bounds have not been shown yet. 
\end{description}
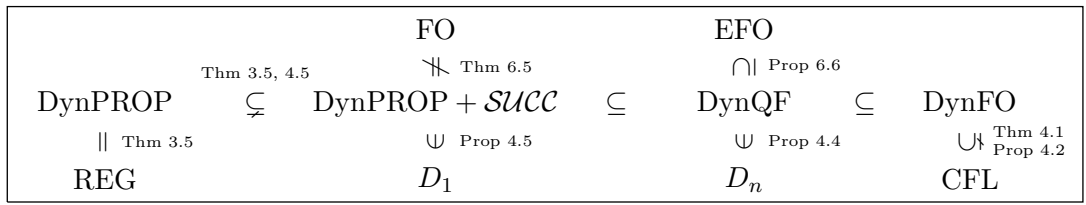
\begin{figure}
\begin{center}
  \begin{pspicture}(.7,0)(15,2)
    \psline[linewidth=.5pt]{-}(.7,-0.3)(15,-.3)(15,2.3)(.7,2.3)(.7,-.3)
    \rput    (2.0,0){REG}
    \rput{90}(2.0,0.5){$=$}
    \rput    (2.7,0.5){\tiny Thm~\ref{theo:dynprop-reg}}
    \rput    (2.0,1.0){$\dynprop$}
    \rput    (4.0,1.0){$\subsetneq$}
    \rput    (4.0,1.4){\tiny Thm~\ref{theo:dynprop-reg}, \ref{theo:dyckone}}
    \rput    (6.4,0){$D_1$}
    \rput{90}(6.4,0.5){$\in$}
    \rput    (7.2,0.5){\tiny Prop~\ref{theo:dyckone}}
    \rput    (6.4,1.0){$\dynprop + \SetSucc$}
    \rput{90}(6.4,1.5){$\neq$}
    \rput    (7.2,1.5){\tiny Thm~\ref{theo:fo-neq-dynprop}}
    \rput    (6.4,2.0){FO}
    \rput    (8.8,1.0){$\subseteq$}
    \rput    (10.5 ,0){$D_n$}
    \rput{90}(10.5 ,0.5){$\in$}
    \rput    (11.3,0.5){\tiny Prop~\ref{theo:dyckn}}
    \rput    (10.5,1.0){$\dynqf$}
    \rput{90}(10.5 ,1.5){$\supseteq$}
    \rput    (11.3,1.5){\tiny Prop~\ref{theorem:EFOdynprop}}
    \rput    (10.5 ,2.0){\EFO}
    \rput    (12.1,1.0){$\subseteq$}
    \rput    (13.5 ,0){CFL}
    \rput{90}(13.5 ,0.5){$\subsetneq$}
    \rput    (14.3,0.5){\begin{minipage}{1cm}\tiny Thm~\ref{theo:cflinfo}\\Prop~\ref{prop:non-cf-func}\end{minipage}}
    \rput    (13.5,1.0){$\dynFO$}
  \end{pspicture}
  \end{center}
  \caption[]{An overview of the main results in this paper.\footnotemark }
  \label{fig:overview}
\end{figure}
\footnotetext{In this figure the dynamic complexity classes are
  allowed to operate with precomputation. Some of the results also hold without precomputation, for example all results concerning formal languages.}
\noindent
Continuing the above lines of research, this paper studies
the dynamic complexity of formal languages with a particular focus on
dynamic classes between \dynprop and \dynQF. Our main contributions
are as follows (see also Figure~\ref{fig:overview}):
\begin{itemize}
\item We give an exact characterization of the dynamic complexity of
  regular languages: a language can be maintained in \dynprop if and
  only if it is regular. This also holds in the presence of arbitrary
  precomputed (aka built-in) relations. (Section~\ref{sec:reg})
\item We provide (presumably) better upper bounds for context-free
  languages: every context-free language can be maintained in \dynFO,
  Dyck languages even in \dynQF, Dyck languages with one kind of
  brackets in a slight extension of \dynprop, where built-in successor
  and predecessor functions can be used. (Section~\ref{sec:cfl})
\item As an immediate consequence, we get a separation between
  \dynprop and \dynQF, thereby also separating \dynprop from \dynFO
  and \dynp.
\item We investigate a slightly different semantic for dynamic string
  languages, and we show that also regular tree languages can be
  maintained in \dynprop, when allowing
  precomputation and the use of built-in functions. (Section~\ref{sec:variations}).
\item We also study general structures, and show that
  (bounded-depth) alternating reachability is not maintainable in
  \dynprop. From this we can conclude that not all first-order
  definable properties are maintainable in \dynprop. On the other hand, we prove
  that all existential first-order definable properties are
  maintainable in \dynQF when allowing
  precomputation. (Section~\ref{sec:graphs})
\end{itemize}

\noindent{\bf Related work.} We already discussed most of the related work
above. A related research area is the study of incremental
computation and the complexity of problems in the cell probe model.
Here, the focus is not on structural (parallel) complexity of
updates but rather on (sequential) update time \cite{Miltersen99cellprobe,MiltersenSVT94}. 
In particular, \cite{FrandsenHMRS95, FrandsenMS97} give efficient incremental
algorithms and analyse the complexity of formal language classes based
on completely different ideas. 

Another area related to dynamic formal languages is the incremental maintenance of schema information (aka regular tree languages) \cite{DBLP:journals/tods/BalminPV04, BarbosaMLMA04} and XPath query evaluation \cite{BjoerklundGMM09} in XML documents. There, the interest is mainly in fast algorithms, less in structural dynamic complexity. Nevertheless techniques of dynamic algorithms on string languages also find applications in these settings.

\section{Definitions}
\label{sec:defs}
Let $\Sigma=\{\sigma_1, ..., \sigma_k\}$ be a fixed alphabet. We
represent words over $\Sigma$ encoded by \emph{word structures}, i.e.,
logical structures $W$ with universe $\{1\ldots,n\}$, one unary relation $R_\sigma$ for each symbol
  $\sigma\in\Sigma$, and the canonical linear order $<$ on  $\{1\ldots,n\}$.
We only consider structures in which, for each $i\le n$, there is at most one $\sigma\in\Sigma$
such that $R_\sigma(i)$ holds, but  there might be none such
$\sigma$. We write $W(i)=\sigma$ if
$R_{\sigma}(i)$ holds and $W(i)=\epsilon$ if no such $\sigma$
exists. We call $n$ the \emph{size} of $W$. 

The word $w=\word(W)$ represented by a word structure $W$ is
simply the concatenation $W(1)\circ\cdots\circ W(n)$. Notice that, due to the fact
that certain elements in $W$ might not carry a symbol,
the actual length of the string can be less than $n$. In particular,
every word $w$ can be encoded by infinitely many different word structures. 
  Let $\intc{i}{j}$ and $\into{i}{j}$ denote the intervals from $i$ to $j$, resp.~from $i+1$ to $j-1$.
  For a word structure $W$, and positions $i\le j$ in $\intc{1}{n}$,
  we write $w\intc{i}{j}$ for the (sub-)string $W(i)\circ\cdots\circ W(j)$. In particular
  $w[i,i-1]$ denotes the empty substring between positions $i$ and $i-1$.
  
By \emptyws we denote the structure with universe $\{1,..,n\}$
representing the empty string $\epsilon$ (thus in $\emptyws$ all
relations $R_\sigma$ are empty).

\subsection{Dynamic Languages and Complexity Classes}

In this section, we first define dynamic counterparts of formal
languages. Informally, a dynamic language consists of all sequences of
insertions and deletions of symbols that transform the empty string
into a string of a particular (static) language $L$. Then we define
dynamic programs which are intended to keep track of whether the
string resulting from a sequence of updates is in $L$. Finally we
define complexity classes of dynamic languages. Most of our
definitions are inspired by \cite{PatnaikI97} but, as we consider
strings as opposed to arbitrary structures,  we try to keep
the formalism as simple as possible.
\bigskip

\noindent{\it Dynamic Languages.}
We will associate with each string language $L$ a dynamic language
$\dynamic{L}$. The idea is that words can be changed by insertions and
deletions of letters and $\dynamic{L}$ is basically the set of update
sequences $\alpha$ which turn the empty string into a string in $L$.

For an alphabet $\Sigma$ we define the set $\ops := \{\ins_\sigma
\mid \sigma\in\Sigma\}\cup\{\reset\}$ of \emph{abstract updates}. 
A \emph{concrete update} is a term of the form $\ins_\sigma(i)$ or
$\reset(i)$, where $i$ is a positive integer. A concrete update is
\emph{applicable} in a word structure of size $n$ if $i\le n$.
By $\Delta_n$ we denote the set of applicable concrete updates for word structures of size $n$.
If there is no danger of confusion we will simply write ``update'' for
concrete or abstract updates. 

The semantics of applicable updates is defined as expected: $\delta(W)$ is the
structure resulting from $W$ by
\begin{itemize}
\item  setting $R_{\sigma}(i)$ to true and $R_{\sigma'}(i)$ to false,
  for $\sigma'\not=\sigma$, if $\delta=\ins_{\sigma}(i)$, and
\item setting all $R_{\sigma}(i)$ to
false, if $\delta=\reset(i)$.  
\end{itemize}
For a sequence $\alpha=\delta_1\ldots
\delta_k\in\Delta_n^+$ of updates we define $\alpha(W)$ as
$\delta_k(\ldots(\delta_1(W))\ldots)$.  

\begin{definition}
Let $L$ be a language over alphabet $\Sigma$. The \emph{dynamic
  language} \dynamic{L}  is the set of all (non-empty) sequences $\alpha$
of updates, for which there is $n > 0$ such that $\alpha \in \Delta_n^+$
and $\word(\alpha(\emptyws)) \in L$. We call $L$ the \emph{underlying language} of
$\dynamic{L}$.\footnote{There is a danger of confusion as we deal with
two kinds of languages: ``normal languages'' consisting of ``normal
strings'' and dynamic languages consisting of sequences of updates. We
use the terms ``word'' and ``string'' only for ``normal strings'' and
call the elements of dynamic languages ``sequences''.}
\end{definition}


\noindent{\it Dynamic Programs.}
Informally, a dynamic program is a transition system which reads
sequences of concrete updates and stores the current string and some
auxiliary relations in its state. It also maintains the information
whether the current string is in the (static) language under consideration.

A \emph{program state} $S$ is a word structure $W$
extended by (auxiliary) relations over the universe of
$W$. The \emph{schema} of $S$ is the set of names and arities of the
auxiliary relations  of $S$. We require that each program has a 0-ary
relation $\accept$.

A \emph{dynamic program} $P$ over alphabet $\Sigma$ and schema $\cR$
consists of an \emph{update function} $\phi^R_{\op}(y;
x_1,\allowbreak \ldots,x_k)$, for every $\op \in \ops$ and $R\in \cR$, where
$k=\arity(R)$.  A dynamic program $P$ operates as follows. Let $S$ be
a program state with word structure $W$. The application of an
applicable update $\delta = \op(i)$ on $S$ yields the new state
$S'=\delta(S)$ consisting of $W'=\delta(W)$ and new relations $R' =
\{\bar{j} \mid S \models \phi^R_{\op}(i,\bar{j})\}$, for each
$R\in\cR$.  For each $n \in \nat$ and update sequence
$\alpha=\delta_1\ldots \delta_k \in \Delta_n^+$ we define $\alpha(S)$
as $\delta_k(\ldots(\delta_1(S))\ldots)$.  We say that a state $S$ is
\emph{accepting} iff $S \models \accept$, i.e., if the 0-ary
\accept-relation contains the empty tuple.\footnote{0-ary relations
  can be viewed as propositional variables: either they contain the
  empty tuple (corresponding to TRUE) or not.}

We say that a dynamic program $P$ \emph{recognizes} the dynamic language
$\dynamic{L}$ if for all $n\in\nat$ and all $\alpha\in\Delta_n^+$ it
holds that $\alpha(\emptyws')$ is accepting iff $\word(\alpha(E_n))\in L$,
where $\emptyws'$
denotes the  state with word structure $\emptyws$ and otherwise empty relations.

\bigskip

\noindent{\it Dynamic Complexity Classes.}
\dynFO is the class of all dynamic languages that are recognized by
dynamic programs whose update functions are
definable by first-order formulas. \dynprop is the subclass of \dynFO
where all these formulas are quantifier free. 

\subsection{Extended Dynamic Programs}
To gain more insight into the subtle mechanics of dynamic
computations, we study two orthogonal extensions of dynamic programs:
auxiliary functions and precomputations.
\bigskip

\noindent{\it Extending dynamic programs with functions.}
A dynamic program $P$ \emph{with auxiliary functions} is a dynamic
program over a schema $\cR$, possibly containing function symbols, which has, 
 for each $\sigma\in\Sigma$ and each function symbol $F\in\cR$ an update function   
	          $\psi^F_\sigma(i;x_1, ..., x_k)$ where $k=\arity(F)$. 

                  As we are mainly interested in quantifier free
                  update functions for updating auxiliary functions we
                  restrict ourselves to update functions defined by
                  \emph{update terms}, defined as:
\begin{itemize}
\item Every $x_i$ is an update term.
\item If $F\in\cR$ is a function and $\bar t$ contains
  only update terms then $F(\bar t)$ is an update term.
\item If $\phi$ is a
  quantifier free formula (possibly using update terms)  and
  $t_1$ and $t_2$ are update terms then
  $\ite(\phi,t_1,t_2)$ is an update term.
\end{itemize}

The semantics of update terms is straightforward for the first two
rules. A term $\ite(\phi,t_1,t_2)$ takes the value of $t_1$
if $\phi$ evaluates to true and the value of $t_2$ otherwise.

After an update $\delta$, the auxiliary functions in the new state
are defined by the update functions in the straightforward way.
Unless otherwise stated, the functions in the initial state
$\emptyws'$ map every tuple to its first element. 
\bigskip

\noindent{\it Extending dynamic programs with precomputations.} 
Sometimes it can be useful for a dynamic algorithm to have a
precomputation which prepares some sophisticated data
structures. Such precomputations can easily be incorporated into the
model of dynamic programs.

In \cite{PatnaikI97} the class $\dynFO^+$ allowed polynomial time
precomputations on the auxiliary relations. The
structual properties of dynamic algorithms with precomputation were
further studied and refined in \cite{WeberS05}. In this paper,
we do not consider different complexities of precomputations but
distinguish only the cases where precomputations are allowed or
not.

A \emph{dynamic program $P$ with precomputations} uses an additional
set of \emph{initial auxiliary relations} (and possibly \emph{initial
  auxiliary functions}). For each initial auxiliary relation symbol
$R$ and each $n$, $P$ has a relation $R_n^{\mbox{init}}$ over
$\{1,\ldots,n\}$. The semantics of dynamic programs with
precomputations is adapted as follows: in the initial state
$\emptyws'$ each initial auxiliary relation $R$ is interpreted by
$R_n^{\mbox{init}}$. Similarly, for initial auxiliary function symbol
$F$ and each $n$ there is a function $F_n^{\mbox{init}}$ over
$\{1,\ldots,n\}$.

Initial auxiliary relations and functions are never updated, i.e., $P$
does not have update functions for them. 
\medskip

The extension of dynamic programs by functions and precomputations can
be combined and gives rise to different complexity classes: For
$I\in\{\bot,\Rel,\Fun\}$ and $A\in\{\Rel,\Fun\}$ we denote by
$\dynC(I,A)$ the class of dynamic languages recognized by dynamic programs
\begin{itemize}
\item without precomputations, if $I=\bot$,
\item with initial auxiliary relations, if $I=\Rel$, 
\item with initial auxiliary relations and functions, if $I=\Fun$,
\item with (updatable) auxiliary relations only, if $A=\Rel$, and
\item with (updatable) auxiliary relations and functions, if $A=\Fun$.
\end{itemize}

Thus, we have $\dynFO=\dynFO(\bot,\Rel)$ and
$\dynprop=\dynprop(\bot,\Rel)$. If the base class $\dynC$ is
$\dynprop$ or $\dynFO$, $\dynC(I,A)$ is clearly monotonic with respect
to the order $\bot<\Rel<\Fun$
In particular, 
\[
\dynprop(\Rel,\Rel)\subs\dynprop(\Fun,\Rel)\subs\dynprop(\Fun,\Fun)
\]
As we are particularly interested in the class $\dynprop(\bot,\Fun)$
we denote it also more consisely by \dynqf. 

As auxiliary functions can be
simulated by auxiliary relations if the update functions are
first-order formulas we also have 
$\dynFO(\Rel,\Rel)=\dynFO(\Fun,\Fun)$ and $\dynFO=\dynFO(\bot,\Fun)$.
Thus, in our setting there are
only two classes with base class $\dynFO$: the one with and the one
without precomputations.
\bigskip

We will also examine the setting where we only allow a specific set of
initial auxiliary (numerical) functions, namely built-in successor and
predecessor functions. For each universe size $n$ let $\succM$ be the
function that maps every universe element to its successor (induced by
the ordering) and the element $n$ to itself, let $\preM$ be the function
mapping to predecessors and the element $1$ to itself, and let $\minM$ be the
constant (i.e.~nullary function) mapping to the minimal element $1$ in
the universe. Then $\dynprop(\SetSucc, \Rel)$ is the class of dynamic
languages recognized by dynamic programs using quantifier free formulas 
with initial (precomputed)
auxiliary relations, the auxiliary functions $\succM$, $\preM$ and $\minM$
and updatable auxiliary relations.
\bigskip

\noindent {\it Dynamic Programs with initialisation.} Let us note here that in some cases dynamic programs need some weak kind of precomputation. In these cases it will be useful to be able to suitably initialize the auxiliary relations, namely in settings where no precomputation is allowed. The following
lemma shows that this is indeed possible, if the initialization
functions can be defined by a quantifier free formula. A \emph{dynamic
program with initialization} is a dynamic program with additional
quantifier free formulas $\beta_R(\bar x)$, for each auxiliary
relation $R$. The value of each relation $R$ in the initial state
$E_n'$ is then determined by $\beta_R$.

\begin{lemma}\label{lem:init}
  For each dynamic \dynprop- or \dynFO-program $P$ with initialization
  there is an equivalent program $P'$ that does not use initialization.
\end{lemma}
\begin{proof}
  The simulating program $P'$ uses an additional 0-ary relation $I_0$ which
  contains the empty tuple if some update has already occurred. The
  update formulas of $P'$ are obtained from those of $P$ by replacing
  each atom of the form $R(\bar x)$ by $(I_0 \land \beta_R(\bar
  x))\lor(\neg I_0 \land R(\bar x))$. The update formulas for $I_0$
  are constantly true.    
\end{proof}

\section{Dynamic Complexity of Regular Languages}\label{sec:reg}

As already mentioned in the introduction, it was  shown by
Patnaik and Immerman \cite{PatnaikI97}  that every regular language
can be recognized by a \dynFO program. Hesse \cite{Hesse03a} showed
that the full power of \dynFO is actually not needed: every regular
language is recognized by some \dynqf program.

Our first result is a precise characterization of the dynamic
languages \dynamic{L} with an underlying regular language $L$: they
exactly constitute the class \dynprop. 
Before stating the result formally and sketch its proof, we will give
a small example to illustrate how regular languages can be maintained
in \dynprop.

\begin{example}\label{example:reg}
  Consider the regular language $(a+b)^*a(a+b)^*$ over the alphabet
  $\{a,b\}$.  One has to maintain one binary relation $A(i,j)$ that is
  true iff $i<j$ and there exists $k\in\into{i}{j}$ such that
  $w\intc{k}{k}=a$ and two unary relations $I(j) \equiv \exists k<j:
  w\intc{k}{k}=a$ and $F(i) \equiv \exists k>i: w\intc{k}{k}=a$. 

We will state here the update formulas for the three kinds of
operations: $\ins_a$, $\ins_b$, and $\reset$. The formulas for the insertion of a $b$ into the string or the deletion of a string symbol are the same, since the language only cares about whether there exist an $a$ in the string or not.
\smallskip

\noindent After the operation $\ins_{a}(y)$, the relations can be updated as follows
  \begin{eqnarray*}
     \phi^{A}_{\ins_a} (y; x_1, x_2) & \equiv & 
       \big[(y \leq x_1 \vee y \geq x_2) \land A(x_1,x_2)\big] \lor \allowbreak \big[x_1 < y < x_2] \\
     \phi^{I}_{\ins_a} (y; x) & \equiv & 
       \big[y \geq x \land I(y)\big] \lor \allowbreak \big[y < x] \\
     \phi^{F}_{\ins_a} (y; x) & \equiv & 
       \big[y \leq x \land F(y)\big] \lor \allowbreak \big[y > x] \\
     \phi_{\ins_a}^{\accept}(y) & \equiv & \text{true},
  \end{eqnarray*}
and after the operations $\ins_{b}(y)$ and $\reset(y)$, the relations can be updated as follows 
  \begin{eqnarray*}
     \phi^{A}_{\reset/\ins_b} (y; x_1, x_2) & \equiv & 
       \big[(y \leq x_1 \vee y \geq x_2) \land A(x_1,x_2)\big] \lor \allowbreak \big[x_1 < y < x_2 \land
       A(x_1,y) \lor  A(y,x_2) \big] \\
     \phi^{I}_{\reset/\ins_b} (y; x) & \equiv & 
       \big[y \geq x \land I(y)\big] \lor \allowbreak \big[y < x \land
       I(y) \lor  A(y,x) \big] \\
     \phi^{F}_{\reset/\ins_b} (y; x) & \equiv & 
       \big[y \leq x \land F(y)\big] \lor \allowbreak \big[y > x \land
       F(y) \lor  A(y,x)\big] \\
     \phi_{\reset/\ins_b}^{\accept}(y) & \equiv & I(y) \lor  F(y).
  \end{eqnarray*}
It is crucial here that $A(i,j)$ refers to the substring from $i+1$ up to
position $j-1$ (as opposed to $i$ and $j$). Otherwise it would not be possible to maintain these auxiliary relations. In the update formula $\phi^{A}_{\ins_{a}} (y; x_1, x_2)$ for example, one 
can only use the three variables $y$, $x_1$ and $x_2$ to compute the new value of $A(x_1, x_2)$ but needs the knowledge about the string on the intervals \into{x_1}{y} and \into{y}{x_2}. \qed
\end{example}

\begin{proposition}\label{prop:regprop}
For every regular language $L$, $\dynamic{L}\in\dynprop$.
\end{proposition}
\proof 
Let $A = (Q,\delta,s,F)$ be a DFA accepting $L$. Here, $Q$ is the set
of states, $\delta: Q \times \Sigma \to Q$ is the transition function,
$s$ is the initial state and $F$ is the set of accepting states. As
usual, we denote by $\delta^*:Q \times \Sigma^* \to Q$ the
reflexive-transitive closure of $\delta$. Then, $w \in L(A)$ iff
$\delta^*(s,w) \in F$.

The program $P$ recognizing $\dynamic{L}$ uses the following relations.
\begin{itemize}
\item  For any pair of states $p,q \in
  Q$, a relation $$R_{p,q} = \{(i,j) \mid i < j \land
  \delta^*(p,w[i+1,j-1]) = q\};$$
\item For each state $q$, a relation $I_q=\{j\mid \delta^*(s,w[1,j-1])=q\}$;
\item For each state $p$, a relation $F_p=\{i\mid
  \delta^*(p,w[i+1,n])\in F\}$,\\
  where $n$ is the size of the word structure.
\end{itemize}

As already mentioned in example \ref{example:reg}, it is crucial here that $R_{p,q}(i,j)$ refers to the substring from position $i+1$ up to position $j-1$ (as opposed to $j$), as will become clear in the following.

Thanks to Lemma \ref{lem:init} we can assume that these relations are
initialized as follows.
\begin{itemize}
\item $R_{p,p} = \{(i,j)
  \mid i < j\}$ and $R_{p,q} = \emptyset$, for $p \neq q$;
\item $I_s=\{1,\ldots,n\}$ and $I_q=\emptyset$, for $q\not=s$;
\item $F_p=\{1,\ldots,n\}$ if $p\in F$ and $F_p=\emptyset$, otherwise.
\end{itemize}

  We now show how these relations can be maintained. First, for each $\sigma \in
  \Sigma$ and $p,q\in Q$, we have the following update formulas for relations $R_{p,q}$
\vspace{-.3cm}
\begin{align*}
  \phi_{\ins_\sigma}^{R_{p,q}}(y;x_1,x_2) & \equiv
  \phantom{\lor\;} \big(y \notin \into{x_1}{x_2} \land R_{p,q}(x_1,x_2)\big) \\
  & \phantom{=\;} \lor\; \big(y \in \into{x_1}{x_2} \land
  \bigvee\limits_{\substack{p',q' \in Q\\ \delta(p',\sigma)=q'}}
  R_{p,p'}(x_1,y)
  \land  R_{q',q}(y,x_2)\big),\\
  \phi_{\reset}^{R_{p,q}}(y;x_1,x_2) & \equiv
  \phantom{\lor\;} \big(y \notin \into{x_1}{x_2} \land R_{p,q}(x_1,x_2)\big) \\
  & \phantom{=\;} \lor\; \big(y \in \into{x_1}{x_2} \land
  \bigvee_{p'\in Q} R_{p,p'}(x_1,y) \land  R_{p',q}(y,x_2)\big).\\
  \intertext{The formulas for the other relations are along the same
    lines, e.g., for each $\sigma \in \Sigma$ and $q\in Q$, and the
    relation $I$ we have the following update formula }
  \phi_{\ins_\sigma}^{I_{q}}(y;x) & \equiv
  \phantom{\lor\;} \big(y \geq x \land I_{q}(x)\big) \\
  & \phantom{=\;}      \lor\;  \big(y < x \land \bigvee\limits_{\substack{p',q' \in Q\\\delta(p',\sigma)=q'}} I_{p'}(y) \land  R_{q',q}(y,x)\big).\\
  \intertext{Finally, $\accept$ can be updated by the formulas}
\end{align*}\vspace{-1.3cm}
\[
\phi_{\ins_\sigma}^{\accept}(y)  \equiv \bigvee\limits_{\substack{p',q' \in Q\\\delta(p',\sigma)=q'}} I_{p'}(y) \land  F_{q'}(y) \qquad\text{ and }\qquad
  \phi_{\reset}^{\accept}(y)  \equiv \bigvee\limits_{p' \in Q}
  I_{p'}(y) \land F_{p'}(y).
\]
\qed


As a matter of fact, the converse of Proposition \ref{prop:regprop} is
also true, thus \dynprop is the exact dynamic counterpart of the
regular languages.

\begin{proposition}\label{prop:propreg}
Let $L=\dynamic{L'}$ be a dynamic language  in $\dynprop$. Then $L'$
is regular.
\end{proposition}

\proof
The idea of the proof is as follows. We consider a dynamic program $P$
for $L$ and see what happens if, starting from the empty word, the
positions of a word are set in a left-to-right fashion. Since the acceptance of
the word by $P$ does not depend on the sequence of updates used to produce the
word, it suffices to consider only this one update sequence.

We make the following observations.
\begin{enumerate}[(1)]
\item After each update, in a sense that will be made precise soon,
  all tuples of positions that have not been set yet behave the same
  with respect to the auxiliary relations.
\item There is only a bounded number (depending only on $P$) of
  possible ways these tuples behave.
\item The change in behavior of the tuples by one update is uniquely determined by
  the inserted symbol.
\end{enumerate}
Together these observations will enable us to define a finite
automaton for $L'$.
\newpage
We first define the concept of the type of a tuple of
elements. Informally, the type of a tuple captures all information a
quantifier free formula can express about a tuple. Let $\bar
i=(i_1,\ldots,i_l)$ be an $l$-tuple of elements of a state $S$ and let
$\varphi$ be a quantifier free formula using variables from
$x_1,\ldots,x_l$. Then we write $\varphi[\bar i]$ for the formula
resulting from $\varphi$ by replacing each $x_j$ with $i_j$. E.g., for
$\bar i=(2,5,4)$ and the atom $\varphi=R(x_3,x_1)$ we get
$\varphi[\bar i]=R(4,2)$.

Let the \emph{type} $\type{S}{\bar i}$ of an $l$-tuple $\bar
i=(i_1,\ldots,i_l)$ in state $S$ be the set of those atomic formulas
$\varphi$ over
$x_1,\ldots,x_l$ for which $\varphi[\bar i]$ holds in $S$. 
A tuple $\bar i=(i_1,\ldots,i_l)$ is \emph{ordered} if
$i_1<i_2<\cdots<i_l$. An \emph{ordered type} is the type of an ordered tuple.

We call a set $I$ of elements of a state $S$ \emph{$l$-indiscernible}
if all ordered $l$-tuples over $I$ have the same type. Notice that if
$l'<l<|I|$ and $I$ is $l$-indiscernible then $I$ is also
$l'$-indiscernible.

Let $P$ be a \dynprop program recognizing a dynamic language
$L=\dynamic{L'}$ and let $k\ge 1$ be the highest arity of any
auxiliary relation of $P$. Our goal is to construct a finite automaton
for $L'$ thus showing that $L'$ is regular. We start by making some
observations.

\paragraph{\bf Observation 1}
{\em  Let $S$ be a state that is reached from $E_n'$ by insertions and
  deletions at positions $\le i$, for some $i$. Then the set
  $\{i+1,\ldots,n\}$ is $k$-indiscernible.}
\proof
Consider two ordered $k$-tuples $\bar j =(j_1, \dots, j_k)$ and $\bar{j'}=(j'_1, \dots, j'_k)$ of elements in $\{i+1,\ldots,n\}$. And let $J$ and $J'$ be the tuples $(1,\dots,i, j_1, \dots, j_k)$ and $(1,\dots,i, j'_1, \dots, j'_k)$.
We will inductively argue that after every considered sequence of updates starting in state $E_n$ and resulting in state $S$ it holds that \[
\type{S}{J} = \type{S}{J'}. 
\]
If this holds for every pair of ordered $k$-tuples in
$\{i+1,\ldots,n\}$ one can conclude that $\{i+1,\ldots,n\}$ is indeed
$k$-indiscernible. Obviously in the state $E_n'$ the equation is true.
Assume now that in some state $S$ the equation holds and consider one
update operation on an element $i'$ in the set $\{1,\ldots,i\}$
resulting in state $S'$. Let $\varphi$ be any atom over the set of
variables $\{x_1, ...,x_{i+k}\}$. The value of $\varphi$ after the
update operation is computed via a quantifier free formula $\psi$ over
$i'$ and $\{x_1, ...,x_{i+k}\}$. Since it holds that $\psi[J]$ is true iff
$\psi[J']$ is true it follows that after the update $\varphi[J]$ is
true iff $\varphi[J']$ is true.
\qed
\paragraph{\bf Observation 2}
{\em Let $S$ be a state and let $l>k$. If a set $I$ of at least $l$
  elements from $S$ is $k$-indiscernible then it is also
  $l$-indiscernible. Furthermore, the type of any ordered $l$-tuple over
  $I$ is uniquely determined by the type of its first $k$ elements.}
\begin{myproof}
  Suppose $I$ is $k$-indiscernible. Let $\bar{i} = (i_1,\ldots,i_l)$
  and $\bar{i'} = (i_1',\ldots,i_l')$ be two ordered $l$-tuples over
  $I$. We show that $\type{S}{\bar{i}} = \type{S}{\bar{i'}}$, from which it then
  follows that $I$ is $l$-indiscernible. To show $\type{S}{\bar{i}} =
  \type{S}{\bar{i'}}$ it suffices to show that for any $R \in S$, with
  $\arity(R) = k'$, and any $j_1,\ldots,j_{k'} \in [1,l]$ it holds that
  $R(i_{j_1},\ldots,i_{j_k'})$ holds in $S$ iff
  $R(i_{j_1}',\ldots,i_{j_k'}')$ holds in $S$. This, however,
  immediately follows from the fact that $k' \leq k$, $I$ is
  $k$-indiscernible, and hence $\type{S}{i_{j_1},\ldots,i_{j_k'}} =
  \type{S}{i_{j_1}',\ldots,i_{j_k'}'}$. Therefore, $\type{S}{\bar{i}} =
  \type{S}{\bar{i'}}$ and thus $I$ is $l$-indiscernable.

  We next show that the type of an ordered $l$-tuple $\bar{i} =
  (i_1,\ldots,i_l)$ is already completely defined by the type of its
  first $k$ elements $i_1$ to $i_k$. Indeed, the type of $\bar{i}$ is
  completely defined by determining, for every relation $R$, with
  $\arity(R) = k'$, and $j_1,\ldots,j_{k'} \in [1,l]$ whether
  $R(i_{j_1},\ldots,i_{j_k'})$ holds in $S$. However, as $k' \leq k$,
  the set $\{i_{j_1},\ldots,i_{j_k'}\}$ contains less than $k$
  different elements and hence as $I$ is $k$-indiscernable, we can
  determine whether $R(i_{j_1},\ldots,i_{j_k'})$ holds in $S$ by
  looking at $\type{S}{i_1,\ldots,i_k}$, the type of its first $k$
  elements.
\end{myproof}

Clearly, the number of possible different
$k$-types is bounded by a number only depending on the schema of $P$.

\paragraph{\bf Observation 3}
{\em Let $S_1,S_1'$ be states with universes of size $n$ and  $n'$,
   respectively and assume that
   $\type{S_1}{i,\ldots,i+l}=\type{S'_1}{i',\ldots,i'+l}$. 
 Let $S_2$ and $S_2'$ be the states resulting
   from  $S_1$ and $S_1'$ by inserting the same symbol $\sigma$ at
   positions $i$ and $i'$, respectively. Then
   $\type{S_1}{i+1,\ldots,i+l}=\type{S'_1}{i'+1,\ldots,i'+l}$.  
}

This observation can be proved along the same lines as the proof of
Observation~1. \medskip

The automaton for $L'$ now is defined as follows.  We call a type
$\tau$ of ordered $k$-tuples \emph{allowed} if there is a (not
necessarily reachable) state $S$ with elements $1,\ldots,k+1$ for
which every ordered $k$-tuple is of type $\tau$.  Let $Q$ be the set
of allowed types of ordered $k$-tuples. For each such type $\tau$ and
each symbol $\sigma$ let $\delta(\tau,\sigma)$ be determined as
follows: Let $S$ be a state\footnote{The states of $P$ should not be
  confused with the states of $A$. We reserve the word ''state'' for
  the former and refer to the latter as $A$-states.} with elements
$\bar i=1,\ldots,k+1$ in which every ordered $k$-tuple is of type
$\tau$. Let $S'$ be the state reached from $S$ after the update
$\ins_\sigma(1)$. Then $\delta(\tau,\sigma)$ is
$\type{S'}{2,\ldots,k+1}$. This new type is also allowed, which can be
seen as follows. Because $\tau$ is an allowed type, the set $\{1,
\dots, k+1\}$ was $k$-indiscernable before the update, and hence
$k'$-indiscernable for any $k' \leq k$. Therefore also the set
$\{2,\dots,k+1\}$ has to be $k'$-indiscernable, for any $k' \leq k$
after the update operation. Now we can add one more element $k+2$ and
define the auxiliary relations of all tuples containing $k+2$ just
like any arbitrary other tuple (not containing $k+2$) with the same
ordering on the elements.  Let $F$ be the set of types for which
$\accept$ holds. Then $A=(Q,\delta,\tau_0,F)$, where $\tau_0$ is
$\type{E'_k}{1,\ldots,k}$. Notice that as the number of $k$-types is
bounded, $A$ is indeed a finite automaton.

We now argue that $L(A) = L'$. Thereto, consider any word $w =
\sigma_1\cdots\sigma_n$, and the associated update sequence $\alpha_w
= \ins_{\sigma_1}(1)\cdots\ins_{\sigma_n}(n)$. Now, we consider an execution
of $P$ on this update sequence in a universe of size $n+k$. Then,
$\word(\alpha_w(E'_{n+k})) = w$, and hence $\alpha_w(E'_{n+k}) \models
\accept$ iff $w \in L'$. Using the observations above it can now be
shown that, for any $i \in [0,n]$, it holds in state
$\alpha_{w[1,i]}(E'_{n+k})$ that (1) the set $\{i+1,\ldots,n+k\}$ is
$l$-indescernable, for any $l$; and (2) $\delta(w[1,i],\tau_0)$ is exactly
the $k$-type of the set $\{i+1,\ldots,i+k\}$, determining the type of
the entire set $\{i+1,\ldots,n+k\}$. As $\tau \in F$ iff $\accept$
holds in $\tau$, it follows that $w \in L(A)$ iff $\alpha_w(E'_{n+k})
\models \accept$.
\qed


\begin{remark}Proposition \ref{prop:propreg} is a powerful tool for
proving lower bounds as it, of course, shows that, for every non-regular language
$L$, $\dynamic{L}\not\in\dynprop$. 
\end{remark}

The proof of Proposition \ref{prop:propreg} intuitively relies on the
fact that all remaining string positions cannot be distinguished
before they are set. Using a Ramsey argument, this idea can be generalized to the setting with
precomputations, thus showing that (relational) precomputations do not increase the expressive power of \dynprop-programs. 
This fact and the above two propositions can then be combined into the following theorem.

\begin{theorem}\label{theo:dynprop-reg}
  Let $L$ be a language. Then, the following are equivalent:
  
\begin{enumerate}
	\item $L$ is regular;
	\item $\dynamic{L} \in \dynprop$; and
        \item $\dynamic{L} \in \dynprop(\Rel,\Rel)$. 
\end{enumerate}
\end{theorem}

\proof 
The only thing left to prove is that for any language $L'$ such that
$L=\dynamic{L'}$ is recognized by a $\dynprop(\Rel,\Rel)$ program, it
holds that $L'$ is regular.  Thereto, we extend the technique of the
proof of Proposition~\ref{prop:propreg} to also handle \dynprop
programs with precomputations. The proof is a generalization of that
proof by a Ramsey argument.

To this end, let $P$ be a $\dynprop(\Rel,\Rel)$ program recognizing a dynamic language
$L=\dynamic{L'}$ and let $k\ge 1$ be the highest arity of any
auxiliary or initial auxiliary relation of $P$. Again, our goal is to
construct a finite automaton for $L'$ thus showing that $L'$ is regular.

The key to the proof is the following observation.
\paragraph{\bf Observation 1'}
{\em  For each $n$ there is some $m$ such that for every state $S$ over a universe of size $m$
  there is a $k$-indiscernible set $I$ of size $n$.
}
\proof This observation can be proved using a version of Ramsey's
theorem for hypergraphs~\cite{book:RamseyTheory}: \emph{Given a number $c$ of colors
  and a natural number $n$ there exists a number $R_c(n)$ such that if
  the edges of a complete $k$-hypergraph (all edges are of size $k$)
  with $R_c(n)$ vertices are colored with $c$ colors, then it must
  contain a complete sub-$k$-hypergraph with $n$ vertices whose edges
  are all colored with the same color.}

Let $c$ be the number of different ordered $k$-types (which only
depends on the number and arity of the initial auxiliary
relations). Then $m$ can be chosen as $R_c(n)$. Consider a state $S$
over a universe of size $m$. Construct a hypergraph $G$ as follows. As
the vertex set use the set of universe elements and add for every set
of elements of size $k$ a $k$-hyperedge colored with its
$k$-type. This leads to a complete $k$-hypergraph for which the vertex
set of each complete monocolored sub-$k$-hypergraph corresponds to a
$k$-indiscernable set. By Ramsey's theorem, $G$ must contain a
monocolored sub-$k$-hypergraph of size at least $n$ and hence $S$
contains a $k$-indiscernable set $I$ of size $n$. \qed

We only consider computations of $P$ which set the elements of some
$k$-indiscernible set in a left-to-right fashion. The automaton $A$ is
constructed similarly as in the proof of Proposition
\ref{prop:propreg}. Now for every string $w$ of some length $n$ there
is, by Observation 1', an $m$ such that every state over $m$
elements has a $k$-indiscernible set $I = \{i_1,\ldots,i_{n+k}\}$ of
size $n+k$. By considering the left-to-right update sequence $\delta_w
= \ins_{\sigma_1}(i_1)\cdots\ins_{\sigma_n}(i_n)$ which sets the word
$w = \sigma_1\cdots\sigma_n$ on the elements of $I$, in a universe of
size $m$, it is easy to
show that $w\in L'$ if and only if $w$ is accepted by $A$.  \qed

\section{Dynamic Complexity of Context-free Languages}\label{sec:cfl}

In the previous section we have seen that the regular languages are
exactly those languages that can be recognized by a \dynprop
program. In this section, we will study the dynamic complexity of
context-free languages.  We first show that any context-free language can be maintained
in \dynFO. Later on, we exhibit languages that can be maintained in
\dynQF or a weak extension of \dynprop.

\begin{theorem}\label{theo:cflinfo}
  Let $L$ be a context-free language. Then, $\dynamic{L}$ is in \dynFO.
\end{theorem}

\proof 

  Let $L$ be a context-free language defined by grammar $G = (V,S,D)$
  over an alphabet $\Sigma$. Here, $V$ is the set of non-terminals, $S
  \in V$ is the initial non-terminal, and $D$ is the set of derivation
  rules. W.l.o.g. we assume that $G$ is in chomsky normal form,
  i.e. every rule in $D$ is either of the form $U \rightarrow XY$,
  with $X,Y \in V$, $U \rightarrow a$, with $a \in \Sigma$, or $U
  \rightarrow \epsilon$. Further, w.l.o.g., we assume that there is a
  distinguished non-terminal $E \in V$ such that $E \rightarrow
  \varepsilon$ and for all $U \in V$, $U \rightarrow UE$ and $U
  \rightarrow EU$. For $U \in V$, and $w \in (V \cup
  \Sigma)^*$, we denote by $U \rightarrow^* w$ that $w$ can be derived
  from $U$. Then, $L(G) = \{w \mid w \in \Sigma^* \land S \rightarrow^*
  w\}$.

  Our dynamic program $P$ recognizing $L$ will maintain for all $X,Y
  \in V$ the following relation: $$R_{X,Y} = \{(i_1,i_2,j_1,j_2) \mid
  [j_1,j_2] \subseteq [i_1,i_2] \land X \rightarrow^* w[i_1,j_1-1] Y
  w[j_2+1,i_2]\}$$

  Intuitively, $(i_1,i_2,j_1,j_2) \in R_{X,Y}$ implies that, assuming
  $Y \rightarrow^* w[j_1,j_2]$, it follows that $X \rightarrow^*
  w[i_1,i_2]$. Notice also that, due to our assumptions above, we have $X
  \rightarrow w[i_1,j_1-1]w[j_2+1,i_2]$ iff $R_{X,E}(i_1,i_2,j_1,j_2)$.
\medskip

  We will now state the update formulae. For each $\sigma \in \Sigma$, and
  $X,Y \in V$ the update formula for $\phi_{\ins_\sigma}^{R_{X,Y}}(z;x_1,x_2,y_1,y_2)$ is
\[
 {[y_1,y_2]} \subseteq {[x_1,x_2]} \land \phi_1 \land \phi_2 \land \phi_3 
\]
where $\phi_1$, $\phi_2$, and $\phi_3$ are defined according to the
  position of $z$ with respect to the other variables:
\begin{align*}
\phi_1 \equiv\;&
  (z \notin [x_1,x_2] \vee z \in [y_1,y_2]) \land R_{X,Y}(x_1,x_2,y_1,y_2)\\
\intertext{In this situation the truth value
    of $R_{X,Y}(x_1,x_2,y_1,y_2)$ is not modified.}
\phi_2 \equiv\;&
   z \in {[x_1,y_1[}\; \land \\
   \bigvee_{\substack{Z,U,U_1,U_2 \in V \\ Z\rightarrow \sigma, U\rightarrow U_1U_2 \in D}}
    & \begin{minipage}{12cm}
    $\exists u_1,u_2,u_3:\;  u_1 \leq u_2 < u_3 \land u_1,u_2 \in {[x_1,y_1[} \land  u_3 \in [y_2,x_2]\; \land  $
    \medskip\\
    $R_{X,U}(x_1,x_2,u_1,u_3)  \land R_{U_1,Z}(u_1,u_2,z,z) \land R_{U_2,Y}(u_2+1,u_3,y_1,y_2)$
    \end{minipage}\\
\intertext{Here the value of $R_{X,Y}(x_1,x_2,y_1,y_2)$ can be
  modified. Figure~\ref{fig:cflUpdate} illustrates this situation. The
  situation if $z \in {]y_2,x_2]}$ and the corresponding formula for $\phi_3$ is quite alike.}
\end{align*}

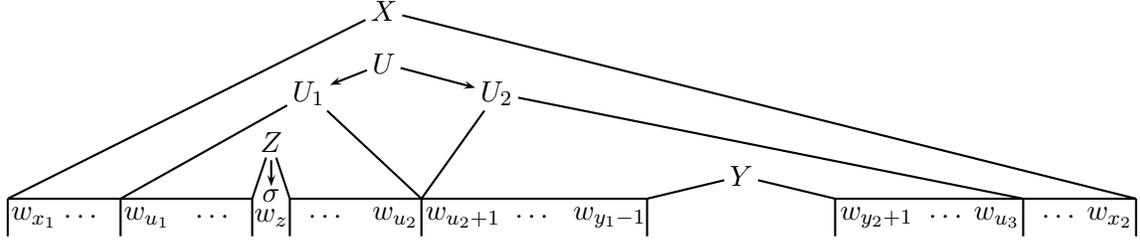
\begin{figure}
\begin{center}
  \begin{pspicture}(15,3.75)
    \pnode(0,1){X1}
    \pnode(1.5,1){U1}
    \pnode(3.25,1){Z1}
    \pnode(3.75,1){Z2}
    \pnode(5.5,1){U2}
    \pnode(8.5,1){Y1}
    \pnode(11,1){Y2}
    \pnode(13.5,1){U3}
    \pnode(15,1){X2}
    \pnode(0,0.5){DX1}
    \pnode(1.5,0.5){DU1}
    \pnode(3.25,0.5){DZ1}
    \pnode(3.75,0.5){DZ2}
    \pnode(5.5,0.5){DU2}
    \pnode(8.5,0.5){DY1}
    \pnode(11,0.5){DY2}
    \pnode(13.5,0.5){DU3}
    \pnode(15,0.5){DX2}
    \ncline{DX1}{X1}
    \ncline{DU1}{U1}
    \ncline{DZ1}{Z1}
    \ncline{DZ2}{Z2}
    \ncline{DU2}{U2}
    \ncline{DY1}{Y1}
    \ncline{DY2}{Y2}
    \ncline{DU3}{U3}
    \ncline{DX2}{X2}
    \rput(5,3.5){\rnode{X}{$X$}}
    \rput(5,2.8){\rnode{U}{$U$}}
    \rput(4,2.4){\rnode{BU1}{$U_1$}}
    \rput(6.5,2.4){\rnode{BU2}{$U_2$}}
    \rput(3.5,1.75){\rnode{Z}{$Z$}}
    \rput(3.5,1.05){\rnode{SIG}{$\sigma$}}
    \rput(9.75,1.3){\rnode{Y}{$Y$}}
    \ncline[nodesepB=2pt]{X1}{X}
    \ncline[nodesepB=2pt]{X2}{X}
    \ncline[nodesepB=2pt]{U1}{BU1}
    \ncline[nodesepB=2pt]{U2}{BU1}
    \ncline[nodesepB=2pt]{U2}{BU2}
    \ncline[nodesepA=3pt, nodesepB=1pt]{->}{Z}{SIG}
    \ncline[nodesepA=2pt]{Y}{Y1}
    \ncline[nodesepA=2pt]{Y}{Y2}
    \ncline[nodesepB=2pt]{U3}{BU2}
    \ncline{X1}{Z1}
    \ncline[nodesepA=2pt]{Z}{Z1}
    \ncline[nodesepA=2pt]{Z}{Z2}
    \ncline{Z2}{Y1}
    \ncline{Y2}{X2}    
    \ncline[nodesep=2pt]{->}{U}{BU1}
    \ncline[nodesep=2pt]{->}{U}{BU2}
    \rput(0.35,0.75){$w_{x_{1}}$}
    \rput(1,0.75){$\cdots$}
    \rput(1.85,0.75){$w_{u_{1}}$}
    \rput(2.75,0.75){$\cdots$}
    \rput(3.5 ,0.75){$w_{z    }$}
    \rput(4.25,0.75){$\cdots$}
    \rput(5.15,0.75){$w_{u_{2}}$}
    \rput(6.05,0.75){$w_{u_{2}+1}$}
    \rput(7,0.75){$\cdots$}
    \rput(8,0.75){$w_{y_{1}-1}$}
    \rput(11.55,0.75){$w_{y_{2}+1}$}
    \rput(12.5,0.75){$\cdots$}
    \rput(13.15,0.75){$w_{u_{3}}$}
    \rput(14,0.75){$\cdots$}
    \rput(14.65,0.75){$w_{x_{2}}$}
  \end{pspicture}
  \end{center}
  \caption{Update of $R_{X,Y}$ after operation $\ins_{\sigma}(z)$}
  \label{fig:cflUpdate}
\end{figure}

  For all $X,Y \in V$, the update formula
  $\phi_\reset^{R_{X,Y}}$ is defined very similar as the
  formula for $\phi_{\ins_\sigma}^{R_{X,Y}}$ above. Essentially, the only
  difference is that $Z$ (for which $Z \rightarrow \sigma \in D$) is
  replaced by $E$ (for which $E \rightarrow \varepsilon \in D$).

  We finally give the update
  formulae for the acceptance relation $\accept$:
\begin{align*}
  \accept_{\ins_\sigma}(z) &\equiv \bigvee_{\substack{Z \in V \\ Z \rightarrow
    \sigma \in D}} R_{S,Z}(\mincons,\maxcons,z,z) \\
\intertext{and}
  \accept_{\reset}(z) &\equiv R_{S,E}(\mincons,\maxcons,z,z).
\end{align*}

Notice that we have used many abbreviations in the above
formulae. However, these can all easily seen to be definable in
first-order logic using the built-in order. In particular, the
constants $\mincons$ and $\maxcons$ and the successor function
are definable and are hence not precomputed functions as
in other settings considered in this paper.
\qed


However, we cannot hope for an equivalence between \dynFO and the
context-free languages, as for \dynprop and the regular languages
before. This follows easily as opposed to the class of context-free
languages, \dynFO is closed under intersection and complement.
Furthermore, one can show that non-contextfree languages can be maintained in $\dynqf$ and $\dynprop(\SetSucc,\Rel)$.
This is because unary counters can be implemented easily by dynamic programs in these classes. 
Let \equaln{r} be the language over
the alphabet $\Sigma = \{a_1,\dots,a_r\}$ containing all strings with an
equal number of occurrences of each symbol $a_i$. Note that already \equaln{3} is not
context-free. Using the counters one can prove the following

\begin{proposition}\label{prop:non-cf-func} \hfill
\begin{enumerate}
	\item  $\dynamic{\equaln{r}} \in \dynprop(\SetSucc,\Rel)$
	\item  $\dynamic{\equaln{r}} \in \dynqf$
\end{enumerate}
\end{proposition}

\proof

In both cases, we just prove the proposition for the language
\equaln{2}. The general case then is an easy generalization of this
proof. 

We will maintain the language \equaln{2} by implementing a unary
counter, which can be done in $\dynprop(\SetSucc,\Rel)$. This counter
will count the difference of the number of occurences of the symbols $a_1$ and
$a_2$ in the string. For $i \in [1,2]$, let $\sharp a_i$ denote the
number of $a_i$s in the current string. We then maintain the following relations:
\begin{itemize}
	\item Nullary relations (flags) $A_1$ and $A_2$ such that
          $A_1$ is true iff $\sharp a_1 > \sharp a_2$ and $A_2$ is
          true iff $\sharp a_2 > \sharp a_1$.
	\item A unary relation $C$ such that $C(i)$ is true iff
          $|\sharp a_1-\sharp a_2|=i$. Hence, as the universe consists
          of the elements $\{1,\ldots,n\}$ at each time $C$ is true
          for one value $i$ if $\sharp a_1 \neq \sharp a_2$ and is
          false for all $i$ iff $\sharp a_1 = \sharp a_2$.
\end{itemize}
We will give the update functions for these relations only for the
case of the insertion of a symbol $a_1$. The deletion and the
$a_2$-case work similarly.

To simplify the presentation we will make the following assumption. We
assume that all update sequences are such that (1) whenever an
update $\reset(z)$ occurs, the position $z$ carried a symbol before
the update, and (2) whenever an update $\ins_\sigma(z)$ occurs, the
position $z$ was empty (i.e. did not carry a symbol). Although a
sequence of updates must not obey these restrictions, it is easy to
transform a program $P$ using these assumptions into an equivalent one
$P'$ which does not. Indeed, for the reset operation, $P'$ can test
whether $z$ used to be empty in which case it can return the original
value of the updated relation or function; or, if $z$ carried a symbol, it can
use the update functions of $P$. In the case of an insertion at a
position $z$ for which $z$ already carried a symbol, $P'$ can
\emph{simulate} what would happen if in $P$ consecutively the updates
$\reset(z)$ and $\ins_\sigma(z)$ would occur. Technically, this can be
achieved by replacing in all formulas $\phi_{\ins_\sigma}^R$ any
occurrence of a relation name $R'$ by $\phi_{\reset}^{R'}$. These
modified update formulas will then compute exactly the relations and
functions $P$ would compute after handling the updates $\reset(z)$ and
$\ins_\sigma(z)$.

Using this assumption, consider the update $\ins_{a_1}(x)$. Then, the
flags $A_1$ and $A_2$ can be updated as follows
\[
\phi_{\ins_{a_1}}^{A_1} \equiv \neg A_2 \text{ and }
\phi_{\ins_{a_1}}^{A_2} = A_2 \wedge \neg C(\minM).
\]
For the update of $C$ we distinguish three cases:
\begin{eqnarray*}
\phi_{\ins_{a_1}}^{C}(x) & \equiv & (\neg(A_1\vee A_2) \wedge x = \minM) \vee\\
			&   & (A_1 \wedge C(\preM(x)) \wedge x \neq \minM) \vee \\
			&   & (A_2 \wedge C(\succM(x))) 
\end{eqnarray*}
The acceptance query just tests whether both $A_1$ and $A_2$ are false
after the update. That is,
$$\phi_{\ins_{a_1}}^{\accept}(x) \equiv \neg \phi_{\ins_{a_1}}^{A_1}(x)
\wedge \neg \phi_{\ins_{a_1}}^{A_2}(x)$$
\bigskip

To proof (2), we will use the same algorithm as before. But, of
course, the algorithm makes extensive use of the functions of
$\SetSucc$, which are not available in $\dynQF$. Instead, we will use
the fact that in $\dynQF$ one can maintain functions to incrementally
construct the $\minM$, $\succM$ and $\preM$ functions.

Here, we do not require that the constructed $\minM$, $\succM$ and
$\preM$ functions are consistent with the order relation. Instead, $\minM$ will
be the first position where a symbol is inserted, its successor the
second such position etc. At each point in time, $\succM$ and $\preM$
therefore define a successor function on those positions that carry a
symbol or carried a symbol earlier. We will not give the precise
update functions which are necessary to construct these auxiliary
functions, but simply mention the ideas necessary to construct them.

Thereto, we additionally maintain a unary relation $\actdom$,
containing all \emph{active} elements currently included in the
successor function, and a constant (i.e. nullary function) $\text{\scshape max}$ denoting the last element of
the successor ordering. Recall that $\succM(\text{\scshape max}) = \text{\scshape max}$ and
$\preM(\minM) = \minM$ should hold by definition of our successor and
predecessor functions.

Then, when an update on an element $x$ occurs there are two
possibilities. Either $\actdom(x)$ already holds in which case nothing
has to be changed, or $\actdom(x)$ does not hold and hence $x$ has to
be added to the successor structure. This is done by setting
$\actdom(x)$, making $x$ the maximal element and setting the
predecessor and successor functions of $x$, $\minM$, and (the old)
$\text{\scshape max}$ corresponding to the new situation.

We finally argue that the program constructed above still works
properly when using these on-the-fly constructed functions instead of
the precomputed ones in $\SetSucc$. Thereto, notice that there are
only two differences. First, the constructed successor functions are
not consistent with the built-in order relation. However, as the
original program does not make use of this order relation, this does
not make a difference. Second, at any time the constructed successor
functions are only defined on $k$ elements, where $k$ is the number of
active elements. However, observe that whenever only $k$ elements are
active, the current string cannot contain more than $k$ symbols, and
hence $C(i)$ does not hold for $i > k$. It should be noted, however,
that $C(k)$ can hold. Therefore we should for every update
\emph{first} compute the new successor functions and use these newly
computed functions in the updates of the other relations. This can
also done without any problems, and hence we can conclude that the original
program remains to work correctly. \qed
 
\noindent From Proposition~\ref{prop:non-cf-func} and Theorem
~\ref{theo:dynprop-reg} one can conclude the following
\begin{corollary}\hfill
\begin{enumerate}
	\item  $\dynprop \subsetneq \dynprop(\SetSucc, \Rel)$
	\item  $\dynprop \subsetneq \dynqf$\qed
\end{enumerate}
\end{corollary}

One can also get better upper bounds for the Dyck-languages,
the languages of properly balanced parentheses. For a set of opening
brackets $\{(_1, ..., (_n\}$ and the set of its closing brackets
$\{)_1, ..., )_n\}$ the language $D_n$ is the language produced by the
context free grammar: 
\[
S \rightarrow SS \mid (_1 S)_1 \mid ... \mid (_n S )_n \mid \epsilon
\]

\begin{proposition}\label{theo:dyckn}
  For every $n>0$, $D_n\in\dynQF$.
\end{proposition}

\proof   The basic idea is similar to the proof of Theorem
  \ref{theo:cflinfo}. We maintain relations 
  $R_1$ and $R_2$ corresponding to $R_{S,E}$ and $R_{S,S}$ in the
  terminology of Theorem \ref{theo:cflinfo}. More precisely,
  $R_1(i_1,i_2)$ should hold if the current substring $w[i_1,i_2]$ is
  well-bracketed. Likewise, $R_2(i_1,i_2,j_1,j_2)$ should hold if the
  string $w[i_1,i_2]$ without the symbols at positions
  $j_1,\ldots,j_2$ is well-bracketed. Stated more formally,
  $R_2(i_1,i_2,j_1,j_2)$ should hold iff $S \rightarrow^*
  w[i_1,j_1-1]\, S\, w[j_2+1,i_2]$.

Nevertheless, the update formulas in the proof of Theorem~\ref{theo:cflinfo} make extensive use of existential quantifiers which
are not available in \dynQF. In
the current proof we will therefore replace these existential quantifiers by
means of functions. To this end, we will maintain several functions
described below. 

As in the proof of Proposition~\ref{prop:non-cf-func}(2), we will make
use of on-the-fly constructed functions $\minM$, $\succM$, and $\preM$,
defined at any time on the elements on which an update already
occurred in the update sequence. Then, we associate numbers with elements in
this successor function, and let $\minM$ denote the number 0, its
successor 1, and so on. We denote the number represented by an element
$v$ as $\repr{v}$. We also denote the element representing a number
$l$ by $\repr{l}$.

Now we can define the four auxiliary functions needed to maintain
$R_1$ and $R_2$. In the following, for two positions $i_1<i_2$, we
write $d(i_1,i_2)$ for the number of closing brackets in
$[i_1,i_2]$ minus the number of opening brackets in
$[i_1,i_2]$. We write $\Cl(v)$ if position $v$ carries a closing
bracket and $\Op(v)$ if it carries an opening bracket.

\begin{itemize}
\item $\fr(u,v)\mydef \minM\{w \mid \repr{v}\ge 1\land \Cl(w)\land w>u \land
  d(u+1,w)=\repr{v}\}$.\\
 Intuitively, $\fr(u,v)$ is the position to
  the right of $u$ where, for the first time, $\repr{v}$ many
  brackets pending at $u$ could be closed.
\item Analogously, $\fl(u,v)\mydef \max\{w \mid \repr{v}\ge 1\land \Op(w)\land w<u \land
  d(w,u-1)=-\repr{v}\}$.
\item $\gr(u,v)\mydef \repr{\max\{d(u+1,w) \mid
  u<w\le v\}}$.\\
 Thus, $\gr(u,v)$ gives the maximum surplus of
  closing brackets in a prefix of $w[u+1,v]$. Intuitively, this is the
  maximum number of pending open brackets at $u$ that can be
  ``digested'' by $w[u+1,v]$. Note that the value of  $\gr(u,v)$ might
  well be $0$.
\item $\gl(u,v)\mydef \repr{\max\{-d(w,u-1) \mid
 v\le w<u\}}$.
\end{itemize}

The attentive reader might have noticed that these functions are not
always defined for all combinations of arguments $u,v$. To this end,
for each of them there is an accompanying relation, telling which function
values are valid. E.g., $R_f^\rightarrow(u,v)$ holds iff $\fr(u,v)$ is
defined. 

As some of the update terms in the dynamic program for $D_n$ are
slightly involved we present the formulas by means of {\em update
  programs} in a pseudocode. These update programs (which should not
be confused with the overall dynamic program) get the parameters of
the relation or function as input, can assign (position) values to
local variables, use conditional branching and return a function value
(or TRUE or FALSE for relations). We abstain from a formal definition
of update programs but it is straightforward to transform them into
update terms by successively replacing each local variable with its
definition.

As noted in the proof of Proposition~\ref{prop:non-cf-func} we can
assume that all update sequences are such that (1) whenever an update
$\reset(z)$ occurs, the position $z$ carried a symbol before the
update, and (2) whenever an update $\ins_\sigma(z)$ occurs, the
position $z$ was empty (i.e. did not carry a symbol).

Using this assumption, we now give the update formulas for the
different relations and functions. In the update programs the
following subroutine $P_0$ will appear three times in update programs
for $D_n$. Its meaning will become clear when it is first used.

\begin{center} 
\begin{myalg}{7cm}{Subroutine $P_0$}
 \STATE \COMMENT{Parameters: $x_1,x_2,y_1,y_2,i_0,j_0,z$}\\
  \STATE $m:= \gr(j_0,y_1-1)$\\
  \STATE $j_1:=\fr(j_0,m)$\\
  \STATE $i_1:=\fl(i_0,m)$\\
  \STATE $m':=\gl(y_1,j_1+1)$\\
  \STATE $j_2:=\fr(y_2,m')$\\
  \IF{$R_1(i_0+1,j_0-1)$ AND\newline
$R_2(i_1,j_1,i_0,j_0)$ AND\newline
 $R_2(j_1+1,j_2,y_1,y_2)$ AND \newline
      $R_2(x_1,x_2,i_1,j_2)$} 
    \STATE Return TRUE
    \ELSE \STATE Return FALSE
  \ENDIF
\end{myalg}
\end{center}

We first give the update program for $R_2(x_1,x_2,y_1,y_2)$ for
insertions of a symbol $(_l$ at a position $z$. Only the case where $z$
is in the left interval (i.e. $[x_1,y_1-1]$) is considered. The other
case is symmetric to the insertion of $)_l$ into the left interval
which will be handled below.

Intuitively, the string is split into four
parts each of which has to be well-bracketed:
\begin{itemize}
\item The string between $i_0=z$ and the corresponding bracket to the
  right ($j_0$) (assuming that this is before $y_1$),
\item the maximally bracketed string (from $i_1$ to $j_1$) around $z$ inside $[x_1,y_1-1]$
  without $[i_0,j_0]$,
\item the substring starting to the right of $j_1$ and ending at the
  corresponding ($=$ matching) position ($j_2$) in $[y_2+1,x_2]$, and
\item the remaining string before $i_1$ and after $j_2$.
\end{itemize}
An illustration can be found in Figure~\ref{fig:strings}(a)

If the matching bracket for $z$ is not before $y_1$ the construction
is slightly different (Figure~\ref{fig:strings}(b)):
\begin{itemize}
\item The string between $z$ and its matching bracket at $j_0$ in
  $[y_2+1,x_2]$ has to be well-bracketed, and
\item the remaining string consisting of $w[x_1,z-1]$ and
  $w[j_0+1,x_2]$ has to be well-bracketed
\end{itemize}

\begin{figure}

  \begin{center}
    
  \begin{pspicture}(12,0.8)
\rput(-1,0.5){(a)}
    \mybput{0}{A}{|}{x_1}
    \mybput{1.7}{B}{[}{i_1}
    \mybput{3}{C}{(}{i_0}
    \mybput{4}{D}{)}{j_0}
    \mybput{5}{E}{]}{j_1}
    \mybput{5.2}{F}{\langle}{}
    \mybput{6}{G}{|}{y_1}
    \mybput{8}{H}{|}{y_2}
    \mybput{9.2}{I}{\rangle}{j_2}
    \mybput{11}{J}{|}{x_2}
   \myline{A}{B}{solid}
   \myline{B}{C}{dotted}
{\psset{dash=2pt 1pt}
   \myline{C}{D}{dashed}}
   \myline{D}{E}{dotted}
   \myline{F}{G}{dashed}
   \myline{H}{I}{dashed}
   \myline{I}{J}{solid}
  \end{pspicture}

\vspace{6mm}

  \begin{pspicture}(12,0.8)
\rput(-1,0.5){(b)}
    \mybput{0}{A}{|}{x_1}
    \mybput{3}{C}{(}{i_0}
    \mybput{6}{G}{|}{y_1}
    \mybput{8}{H}{|}{y_2}
    \mybput{9.2}{D}{)}{j_0}
    \mybput{11}{J}{|}{x_2}
   \myline{A}{C}{solid}
   \myline{C}{G}{dotted}
   \myline{H}{D}{dotted}
   \myline{D}{J}{solid}
  \end{pspicture}
  \end{center}
  \caption{Illustration of the update programs for (a) insertion of
    $($ if the matching bracket is in the left string, (b) if it is in
  the right string. }
  \label{fig:strings}
\end{figure}
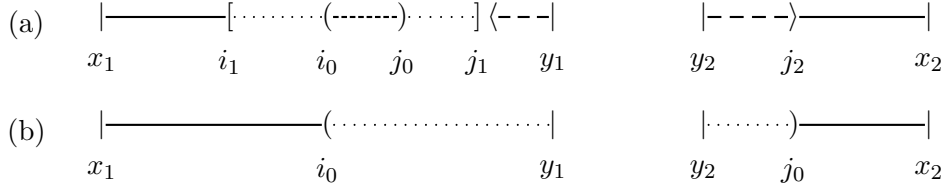

\begin{center} 
\begin{myalg}{7cm}{Update $R_2(x_1,x_2,y_1,y_2)$: insert $(_l$ at $z$}
  \IF{$z\in[x_1,y_1-1]$}
    \STATE $i_0:=z$\\
    \STATE $j_0:=\fr(z,1)$ \COMMENT{find the matching closing bracket}\\
    \IF {$j_0<y_1$}
      \IF{$R_{)_l}(j_0)$}
        \STATE $P_0$
      \ELSE 
        \STATE Return FALSE\\
      \ENDIF
  \ELSE
  \STATE $m:=\gl(y_1,z)$\\
  \STATE $j_0:=\fr(y_2,m+1)$
  \IF{$R_{)_l}(j_0)$ AND\newline
$R_2(z+1,j_0-1,y_1,y_2)$ AND\newline
 $R_2(x_1,x_2,z,j_0)$}  
    \STATE Return TRUE
    \ELSE 
     \STATE Return FALSE
  \ENDIF
  \ENDIF
  \ELSE
   \STATE \COMMENT{Symmetric case $z\in[y_2+1,x_2]$}
  \ENDIF
\end{myalg}
\end{center}

Note that the internal variable $m$ is used for a position that is
interpreted as a number (encoded as explained above). Thus, $m+1$ is
an abbreviation for $\succM(m)$. Likewise, $0$ is an abbreviation for $\minM$.

It could be the case that in line 3 no matching bracket is found. In
this case the update program fails and returns FALSE. In the actual
function terms this can be handled by the help of relation $R_f^\rightarrow$.
We will stick to
this convention also in the following: whenever a function value is
not defined the value of the update program becomes FALSE
(corresponding to undefined values for the function update programs below).

Next we describe the update program for insertions of $)_l$. This case
is very similar to the insertion of $(_l$: The only difference is that
$j_0$ is now the position $z$ and $i_0$ is the matching position to
the left. Furthermore, there is no case distinction as $j_0$ is always
in the left string.

\begin{center} 
\begin{myalg}{7cm}{Update $R_2(x_1,x_2,y_1,y_2)$: insert $)_l$ at $z$}
  \IF{$z\in[x_1,y_1-1]$}
    \STATE $i_0:=\fl(z,1)$\\
    \STATE $j_0:=z$\\
    \STATE $P_0$
  \ELSE
   \STATE \COMMENT{Symmetric case $z\in[y_2+1,x_2]$}
  \ENDIF
\end{myalg}
\end{center}

Finally, the following update program handles reset operations. This
can be handled just as an insertion but here there is no string
between $i_0$ and $j_0$. 

\begin{center} 
\begin{myalg}{7cm}{Update $R_2(x_1,x_2,y_1,y_2)$: reset $z$}
  \IF{$z\in[x_1,y_1-1]$}
    \STATE $i_0:=z$\\
    \STATE $j_0:=z$ \COMMENT{The empty string $w[z+1,z-1]$ is
      well-bracketed...}\\ 
    \STATE $P_0$
  \ELSE
   \STATE \COMMENT{Symmetric case $z\in[y_2+1,x_2]$}
  \ENDIF
\end{myalg}
\end{center}

The update programs for $R_1$ are similar but easier. We now describe
the update programs for the functions $\fl,\fr,\gl,\gr$. We only
describe the update programs for $\fr$ and $\gr$ as $\fl$ and $\gl$
are again symmetric. We do not explicitly state the update programs
for $R_f^\rightarrow$ and $R_g^\rightarrow$ as they are completely
analogous to the programs for the functions.

For $\fr(x,m)$ we only need to consider the case where $m$ has a corresponding
number and is different from $\minM$.

The insertion of $(_l$ at position $z$ only affects $\fr(x,m)$ if
$x<z<\fr(x,m)$. In that case, the insertion of $z$ increases $d(x,w)$
by one for all $w>z$ and therefore the previous value of $\fr(x,m+1)$
is the new value for $\fr(x,m)$. 

\begin{center} 
\begin{myalg}{7cm}{Update $\fr(x,m)$: insert $(_l$ at $z$}
  \IF{$z\le x$}
    \STATE return $\fr(x,m)$
   \ELSE
    \STATE $y:=\fr(x,m)$\\
    \IF {$y<z$}
      \STATE $y:=\fr(x,m)$
    \ELSE
      \STATE $y:=\fr(x,m+1)$
    \ENDIF
  \ENDIF     
\end{myalg}
\end{center}

Notice that in this program we are using the assumption that $z$ was
empty before the insertion. The update of $\fr(x,m)$ under insertion
of a closing bracket is slightly more involved. If $x<z<\fr(x,m-1)$
then the new value is just $\fr(x,m-1)$. Otherwise, we have to
identify the maximal pair of matching brackets around $z$ where the
left bracket is to the right of $\fr(x,m-1)$ ($= y$). Due to the
additional closing bracket at $z$ the right bracket of this pair
($y'$) is then the new value for $\fr(x,m)$. In case $m=1$ we simply
replace the role of $\fr(x,m-1)$ by $x$. The main case is illustrated
by Figure \ref{fig:updatef}

\begin{figure}

  \begin{center}
  \begin{pspicture}(12,0.8)
    \mybput{0}{A}{|}{}
    \mybput{2}{B}{|}{x}
    \mybput{3.5}{C}{]}{}
    \mybput{4.5}{D}{]}{y}
    \mybput{4.7}{E}{\langle}{}
    \mybput{6}{F}{)}{z}
    \mybput{7}{G}{\rangle}{y'}
    \mybput{11}{H}{|}{}
   \myline{A}{B}{solid}
   \myline{B}{C}{dashed}
   \myline{C}{D}{dashed}
   \myline{E}{F}{dotted}
   \myline{F}{G}{dotted}
   \myline{G}{H}{solid}
  \end{pspicture}
  \end{center}
  \caption{Illustration of the update program for $\fr$ under insertion of
    a closing bracket. }
  \label{fig:updatef}
\end{figure}
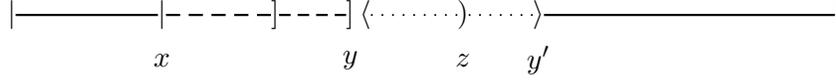

\begin{center} 
\begin{myalg}{7cm}{Update $\fr(x,m)$: insert $)_l$ at $z$}
   \IF{$z\le x$}
    \STATE Return $\fr(x,m)$
   \ELSE
    \IF{$m=1$}
     \STATE $y:=x$
    \ELSE
     \STATE $y:=\fr(x,m-1)$\\
    \ENDIF
    \IF {$y>z$}
     \STATE Return $y$
    \ELSE
     \STATE $m':=\gl(z,y+1)$\\
     \IF {$m'=0$}
        \STATE Return $z$
     \ELSE
     \STATE Return $\fr(z,m')$\\
     \ENDIF
    \ENDIF
   \ENDIF     
\end{myalg}
\end{center}

The update program for a reset operation is similar to the insertion
of $)_l$ in case $z$ carries an opening bracket and simple if $z$ carries a closing bracket.

\begin{center} 
\begin{myalg}{7cm}{Update $\fr(x,m)$: reset $z$}
   \IF{$z$ carries a closing bracket}
     \STATE $y:=\fr(x,m)$
     \IF{$y<z$}
      \STATE Return $y$
     \ELSE  \STATE Return $\fr(x,m+1)$
     \ENDIF
   \ELSE
   \IF{$z\le x$}
    \STATE return $\fr(x,m)$
   \ELSE
    \IF{$m=1$}
     \STATE $y:=x$
    \ELSE
     \STATE $y:=\fr(x,m-1)$\\
    \ENDIF
    \IF {$y>z$}
     \STATE Return $y$
    \ELSE
     \STATE $m':=\gl(z,y+1)$\\
     \STATE Return $\fr(z,m'+1)$\\
    \ENDIF
   \ENDIF     
   \ENDIF
\end{myalg}
\end{center}

Next, we give update programs for $\gr(x,y)$. The first one handles
insertion of an opening bracket and also the reset for closing brackets.

\begin{center} 
\begin{myalg}{7cm}{Update $\gr(x,y)$: insert $(_l$ at $z$}
   \IF{$z\le x$ OR $z>y$}
    \STATE Return $\gr(x,y)$
   \ENDIF
     \STATE $m:=\gr(x,y)$
     \STATE $v:=\fr(x,m)$
    \IF {$v<z$}
     \STATE Return $m$
    \ELSE
     \STATE Return $m-1$
    \ENDIF
\end{myalg}
\end{center}

The next one handles insertion of closing brackets.

\begin{center} 
\begin{myalg}{7cm}{Update $\gr(x,y)$: insert $)_l$ at $z$}
   \IF{$z\le x$ OR $z>y$}
    \STATE Return $\gr(x,y)$
   \ENDIF
     \STATE $m:=\gr(x,y)$
     \STATE $v:=\fr(x,m)$
    \IF {$v>z$}
     \STATE Return $m+1$
    \ENDIF
     \STATE $m':=\gl(z,v)$
    \IF{$m'=0$}
     \STATE Return $m+1$
    \ENDIF
    \IF{$\fr(z,m')\le y$}
     \STATE Return $m+1$
    \ELSE
     \STATE Return $m$
    \ENDIF
\end{myalg}
\end{center}

The last update program takes care of reset of opening brackets.

\begin{center} 
\begin{myalg}{7cm}{Update $\gr(x,y)$: reset $(_l$ at $z$}
   \IF{$z\le x$ OR $z>y$}
    \STATE Return $\gr(x,y)$
   \ENDIF
     \STATE $m:=\gr(x,y)$
     \STATE $v:=\fr(x,m)$
    \IF {$v>z$}
     \STATE Return $m+1$
    \ENDIF
     \STATE $m':=\gl(z,v)$
    \IF{$\fr(z,m'+1)\le y$}
     \STATE Return $m+1$
    \ELSE
     \STATE Return $m$
    \ENDIF
\end{myalg}
\end{center}

Finally, we give the update formulas for the acceptance relation
$\accept$. Thereto, we maintain two additional constants (0-ary
functions) $\first$ and $\last$. Here, $\first$ will denote the first
element (first according to the given order, not the constructed
successor functions) which has been touched, and, similarly, $\last$
denotes the last such element. Hence, at any time $w[1,\first-1] =
w[\last+1,n] = \varepsilon$. These functions can easily be maintained. We
give the update formulas for our acceptance relation again in our
usual formalism:
$$\phi_{\ins_\sigma}^\accept(z) \equiv
\phi_{\ins_\sigma}^{R_1}(z;\phi_{\ins_\sigma}^\first(z),\phi_{\ins_\sigma}^\last(z))$$
and
$$\phi_{\reset}^\accept(z) \equiv
\phi_{\reset}^{R_1}(z;\phi_{\reset}^\first(z),\phi_{\reset}^\last(z)).$$

That is, the string is valid iff $R_1(\first,\last)$ holds after the update
has occurred. This completes the description of the update
programs. The correctness proof is tedious but straightforward.
\qed

We expect the result to hold for a broader class of
context-free languages which has yet to be pinned down exactly. It is
even conceivable that all deterministic or unambiguous context-free
languages are in \dynQF.

It turns out that for Dyck languages with only one kind of brackets,
i.e., $D_1$, auxiliary functions are not needed, if built-in successor
and predecessor functions are given. 

\begin{proposition} \label{theo:dyckone}
$D_1\in\dynprop(\SetSucc, \Rel)$ 
\end{proposition}

\proof In \cite{PatnaikI97} it was shown that $D_1$ is maintainable in \dynFO
using the well known "level trick". To each position $i$ of the string
a number $L(i)$ (the level) is assigned such that $L(i)$ is equal to
the number of opening brackets minus the number of closing brackets in
the substring $w[1..i]$. Then the string is in $D_1$ iff there is no
negative level and the level of the last position in the string equals
0.

In the following program we will maintain a data structure, called a
ringlist, capable of storing a set of elements. Here, a \emph{ringlist} is
the edge relation of a directed graph that is a circle. For instance,
the set $\{a,b,c\}$ can be stored by storing the edge relation
$\{(a,b),(b,c),(c,a)\}$.

The $\dynprop(\SetSucc, \Rel)$-program for $D_1$ will maintain for all
pairs $(i,j)$ of positions in the string and for each number
$l\in\{-n, ..., -2, -1, 0, 1, 2, ..., n\}$ a ringlist of all positions
$k\in\{i,..j\}$ of level $l$. Thereto we will use the following
relations:
\begin{itemize}
  \item $L_0(i,j,\cdot,\cdot)$ is a 4-ary relations containing the ringlist of all string positions of level 0,
	\item $L_+(i,j,l,\cdot,\cdot)$ and $L_-(i,j,l,\cdot,\cdot)$ are 5-ary relations containing ringlists for the positive and negative level $l$ and $-l$.
	\item $F_0(i,j)$ is binary and holds if $L_0(i,j)$ is not empty.
	\item $F_+(i,j,\cdot)$, $F_-(i,j,\cdot)$ are 3-ary relations
          telling whether the corresponding lists are not empty.
	\item $Fmax_0(i)$ as the unary relation that will be equal to
          $F_0(i,n)$ where $n$ is the universe size (remember that we
          only have access to the minimal element). 
	\item $Fmax_-(i,l)$ and $Fmax_+(i,l)$ equal to $F_-(i,n,l)$ and $F_+(i,n,l)$. 
	\item $Min_0(i,j,k)$ ($Max_0(i,j,k)$) is ternary and will be true iff $k$ is the minimal (maximal) element of the ringlist $L_0(i,j)$. 
	\item $Min_+(i,j,l,k)$, $Min_-(i,j,l,k)$, $Max_+(i,j,l,k)$ and $Max_-(i,j,l,k)$ are the corresponding relations for the ringlist of the other levels beside 0.
	\item $Last_0$ is a nullary relation stating that the level of the last position is 0. 
	\item $Last_-(l)$ and $Last_+(l)$ store the level of the last position.
\end{itemize}
\medskip

Initially we have that for all $i$ and $j$, $F_0(i,j)$, $Fmax_0(i)$, $Last_0$, $Min_0(i,j,i)$ and $Max_0(i,j,j)$ are true and 
\[
L_0(i,j,a,b) = (a,b\in\{i,..,j\} \land b=\succM(a)) \vee (a=j \wedge b=i).
\]
Thanks to Lemma \ref{lem:init} we can assume these initializations to take place before the computation of the program.
\medskip
 
We can maintain these relations because of the following observation:
{\em After an update operation on some position $x$ in the string, the
  level of all succeeding positions increases or decreases
  simultaniously by 1.} 

Here again (like in the proof of
Proposition~\ref{prop:non-cf-func}) we can assume that (1) whenever an
update $\reset(z)$ occurs, the position $z$ carried a symbol before
the update, and (2) whenever an update $\ins_\sigma(z)$ occurs, the
position $z$ was empty (i.e. did not carry a symbol).

So to get the new ringlist for some level $l$ after an update at a
position $x$ one has to merge the ringlist for the position between
$i$ and $\preM(x)$ of level $l$ and the one for position between $x$
and $j$ of level $l+1$ or $l-1$. In order to do this, only the
relations around the update position $x$, its two borders $i$ and $j$
and the minimal and maximal element (relative to the ordering) of the
considered ringlistes have to be changed. We will show that it
possible to express these updates using quantifier free formulas.
\medskip

Let us first consider the update function for $L_0(i,j)$ and the
operation $\ins_((x)$. Here the levels of all positions from $x$ to
$n$ have to increase by one. The update formulas for the relations
$L_-$ and $L_+$ are then along the same line, and also the ones for the
update operations $\ins_)(x)$ and $\reset(x)$ can be obtained in the same
way.  For readability we will
use case distinctions and state the formulae for each case
separately. They can easily be put together in one (quantifier free)
formula.
\begin{itemize}
\item If $x$ does not lie in the interval $[i,j]$ then nothing
  happens, $L_0$ remains the same.
\item If $x=i$ then the whole list has to be increased by one,
  so \[\phi^{L_0}_{\ins_(}(x;i,j,a,b) \equiv L_-(i,j,1,a,b).\] Let us
  remark here that the constant 1 is not included as a nullary
  function but can be accessed via $\succM(\minM)$.
\item Else, if $x\in [\succM(i), j]$ then one has to merge the list
  $L_0(i,\preM(x),\cdot,\cdot)$ and $L_-(x,j,1,\cdot,\cdot)$. Here the
  emptiness-relations $F_0(i,\preM(x))$ and $F_-(x,j,1)$ come into
  play, because if one of the corresponding ringlists is empty, the
  other ringlist just has to be copied. If both are empty, then
  $L_0(i,j,\cdot,\cdot)$ has to be empty after the update. Only if
  both $F_0(i,\preM(x))$ and $F_-(x,j,1)$ are false the following
  formula applies:
\begin{eqnarray*}
   \phi^{L_0}_{\ins_(}(x;i,j,a,b) & \equiv & a<b<x \wedge L_0(i,\preM(x),a,b) \vee \\
                 &   & a<x\leq b \wedge (Max_0(i,\preM(x),a) \wedge Min_-(x,j,1,b)) \vee \\
                 &   & x\leq a < b \wedge L_-(x,j,1,a,b) \vee \\
                 &   & b<a \wedge (Min_0(i,\preM(x), b) \wedge Max_-(x,j,1,a))
\end{eqnarray*}
\end{itemize}

Similar to the level-relations $L_0$, $L_-$ and $L_+$ we will state
here only the update formula for $F_0(i,j)$ after the update operation
$\ins_((x)$. The formulas for the other emptiness-relations $F_-$ and
$F_+$ and for the other kind of update operations are similar. 
\begin{itemize}
\item If $x$ does not belong to $\intc{i}{j}$, then $F_0(i,j)$ stays the same;
\item if $x=i$ then \[\phi^{F_0}_{\ins_(}(x;i,j)\equiv F_-(i,j,1)\] because
  the whole ringlist $L_-(i,j,1,\cdot,\cdot)$ was shifted to
  $L_0$. Hence, if $L_-$ was empty before the update operation then
  after the update $L_0$ should be empty.
\item
In the third case, if $x\in [\succM(i), j]$ then $F_0(i,j)$, is
non-empty if either $L_0(i,\preM(x), \cdot, \cdot)$ or
$L_-(x,j,1,\cdot,\cdot)$ was non-empty before the update operation. So
\[
\phi^{F_0}_{\ins_(}(x;i,j) \equiv F_0(i,\preM(x)) \lor F_-(x,j,1).
\]
\end{itemize}
The relations $Fmax_0$, $Fmax_-$ and $Fmax_+$ can be maintained in a similar way.
\medskip

Now we will show how to maintain the relation $Min_0(i,j,k)$ after the
update operation $\ins_((x)$. Again, three cases have to be
distinguished.
\begin{itemize}
\item 
If $x\notin\intc{i}{j}$ then nothing changes.
\item If $x=i$ then \[\phi^{Min_0}_{\ins_(}(x;i,j,k) \equiv Min_-(i,j,1,k)\]
\item Else, we have to check whether the list
  $L_0(i,\preM(x),\cdot,\cdot)$ is empty or not. If it is empty, then
  the minimum has to be taken from the list $L_-(x,j,1)$. If not its minimum
  remains the same. So we have the following formula for the third case:
\begin{eqnarray*}
\phi^{Min_0}_{\ins_(}(x;i,j,k) & \equiv & (\phantom{\neg}F_0(i,\preM(x)) \land Min_0(i,\preM(x),k)) \lor \\
															 &   & (\neg F_0(i,\preM(x)) \land Min_-(x,j,1,k))
\end{eqnarray*}
\end{itemize}

Again, the relations $Max_0$, $Min_-$, $Min_+$, $Max_+$ and $Max_-$
can be updated similarly. The last relations which have to be updated
are $Last_0$, $Last_-$ and
$Last_+$. However, their change does not depend on the position of the actual
update operation, but only on the type of the inserted or deleted
symbol. In fact they only have to count the difference between the
number of opening and closing brackets in the string. The maintenance
of these relations is straightforward. For example after the insertion of an
opening bracket we have $\phi^{Last_0}_{\ins_(}(x)\equiv Last_-(1)$.
\bigskip

Now only the acceptance of a string has to be detected. The string
will be accepted if and only if, after the update, the level of the
last position equals 0 and the ringlist of level -1 is empty. We only
have to check the level -1, and not all negative levels, because if
there is a position with level less then -1 there also has to be a
position which has level -1. So, for instance for the update $\ins_($, the
update formula for \accept is
\[
\phi^{\accept}_{\ins_(}(x) \equiv \neg \phi^{Fmax_-}_{\ins_(}(x;\minM, 1) \land \phi^{Last_0}_{\ins_(}(x)
\]

\qed

So we see that, whereas built-in relations did not increase the power of \dynprop,
already the three simple functions  $\succM$, $\preM$ and $\minM$
allow the maintenance of non-regular languages.

\section{Variations}\label{sec:variations}

\noindent
{\it Alternative Semantics.}
Following~\cite{PatnaikI97}, we have introduced in
Section~\ref{sec:defs} dynamic languages in which it is both allowed
to insert or change labels at positions in the string and to delete
elements at positions. In a universe of size $n$, one can thus create all
strings of length smaller or equal than $n$.

However, one can also consider the setting in which each position in
the string must at any time be assigned a symbol. Although this
setting is less ``dynamic'', it has the advantage that a word is
always associated with its canonical logical structure.  This can be
achieved by starting with an initial structure in which each symbol is
already assigned a symbol, and subsequently only allowing labels to be
changed (and not deleted).

More formally, we assign to every language $L$, a dynamic language
$\dynamicAlt{L}$ as follows. For a distinguished \emph{initial symbol}
$a \in \Sigma$, and $n \in \nat$, let $E^a_n$ be the word structure in
which $R_a(i)$ is true, for all $i$, and $R_\sigma$ is empty, for all
$\sigma \neq a$. Further, $\Delta_n = \{\ins_\sigma \mid \sigma \in
\Sigma\}$. Then, $\dynamicAlt{L} = \{(n,\delta) \mid \delta \in
\Delta_n^+ \wedge \word(\delta(E^a_n)) \in L\}$\footnote{Notice that
  $\dynamic{L}$ consists only of update sequences $\delta$, whereas
  $\dynamicAlt{L}$ contains tuples $(n,\delta)$. This change is
  necessary as the membership of a word of a language under the
  current semantics can depend both on the size of the initial
  structure $n$, and the update sequence $\delta$.}.

Proposition~\ref{prop:alt-sem} shows that the situation is less
appealing than in the original semantics. In particular, there are
regular languages which cannot be maintained without precomputation;
and with precomputation all regular, but also non-regular, languages
can be maintained. Here, $\midd = \{wbw' \mid |w| = |w'|\}$ is the
language over the alphabet $\Sigma = \{a,b\}$ which contains all
strings whose middle element is $b$, which is clearly not regular.

\begin{proposition}\label{prop:alt-sem}\hfill
  \begin{enumerate}
  \item $\dynamicAlt{L((aa)^*)} \notin \dynprop$
  \item For any regular language $L$, $\dynamicAlt{L} \in \dynprop(\Rel,\Rel)$
  \item $\dynamicAlt{\midd} \in \dynprop(\Rel,\Rel)$
  \end{enumerate}
\end{proposition}
\newpage
\proof
  (1) Let $L = L((aa)^*)$. Let $n$ be any positive even integer, and
  $\delta = \ins_a(1)$. Then, $\word(\delta(E^a_n)) \in L$, and
  $\word(\delta(E^a_{n+1})) \notin L$. Hence, $(n,\delta) \in
  \dynamicAlt{L}$ and $(n+1,\delta) \notin \dynamicAlt{L}$. We show
  that for any program $P \in \dynprop$, $(n,\delta) \in L(P)$
  iff $(n+1,\delta) \in L(P)$, which implies the proposition.

  Thereto, notice that $(n,\delta) \in L(P)$ iff $E'^a_n \models
  \phi_{\ins_a}^\accept(1)$, and, correspondingly, $(n+1,\delta) \in L(P)$ iff
  $E'^a_{n+1} \models \phi_{\ins_a}^\accept(1)$. However, these two
  questions can be decided in an identical manner: take
  $\phi_{\ins_a}^\accept$, replace any occurrence of $R_a$ by true and
  any occurence of a relation symbol different from $R_a$ by false,
  and evaluate the obtained boolean formula. Hence, $E'^a_{n} \models
  \phi_{\ins_a}^\accept(1)$ iff $E'^a_{n+1} \models
  \phi_{\ins_a}^\accept(1)$, which concludes the proof.

\smallskip

\noindent (2) As seen in the previous proof, 
$\dynprop$ program without precomputation are not capable of
maintaining all regular languages. The reason for this is that the
initial string is $a^n$, for some $n$, whereas the initial string was empty
in the original semantics. Then, when the computation starts, the
\dynprop program did not have the chance to initialize its data
structures according to $a^n$ and is immediately lost.

However, when allowing precomputation, we can simply reuse the program $P$
defined in the proof of Proposition~\ref{prop:regprop}. Indeed, the only
difference is in the initialization of the relations. Whereas they
could be initialized by quantifier free formulas when the initial
string was empty, we now have to use the power of precomputations
to initialize them. In particular, for a language $L$ accepted by
automaton $A = (Q,\delta,s,F)$ they should be initialized as follows:
\begin{itemize}
\item $R_{p,q} = \{(i,j) \mid i < j \wedge
  (p,a^{j-i-1},q) \in \delta\}$;
\item $I_q=\{i \mid (s,a^{i-1},q) \in \delta\}$; and
\item $F_p=\{i \mid (p,a^{n-i},q_f) \in \delta, \mbox{ for some }
  q_f \in F\}.$
\end{itemize}

From the correctness of the program of Proposition~\ref{prop:regprop}
and this precomputation, the correctness of this modified program
immediately follows.

\smallskip

\noindent (3) The dynamic program $P$ maintaining $\dynamicAlt{\midd}$
will make use of the precomputed unary relation $M$ containing the
middle element of the structure, if the universe size is
odd. Formally, for $n \in \nat$, $M^\text{init}_n = \{\lceil n/2
\rceil \mid n \mbox{ is odd}\}$. Then, $P$ only needs to maintain the
acceptance relation, which can be done as follows:

$$\phi_{\ins_a}^\accept(x) \equiv \accept \wedge \neg M(x)$$
and
$$\phi_{\ins_b}^\accept(x) \equiv \accept \vee M(x).$$
\qed

Notice that, contrary to Theorem~\ref{theo:dynprop-reg}, Proposition~\ref{prop:alt-sem} 
does not allow us to infer lower bounds for
$\dynprop(\Rel,\Rel)$ under the current semantics. However, if we
consider the class of languages with neutral elements, this becomes
possible again. We say that a language $L$ has a \emph{neutral element} $a$
if for all $w,w' \in \Sigma^*$ it holds that $ww' \in L$ iff $waw'
\in L$. Here, if a language has at least one neutral element we will assume
that the initial symbol for its dynamic algorithm is one of these
neutral elements.

Then, a straigthforward generalization of
Theorem~\ref{theo:dynprop-reg} yields the following proposition which
implies, for instance, that $\dynamicAlt{L} \notin \dynprop(\Rel,\Rel)$
for all non-regular languages $L$ which have a neutral element.

\begin{proposition}\label{prop:dynpropalt-reg}
  Let $L$ be a language which has a neutral element. Then, the
  following are equivalent:
    
\begin{enumerate}
	\item $L$ is regular;
	\item $\dynamicAlt{L} \in \dynprop$; and
        \item $\dynamicAlt{L} \in \dynprop(\Rel,\Rel)$. 
\end{enumerate}
\end{proposition}

\proof
  As $(2) \Rightarrow (3)$ follows by definition, it suffices to show
  $(1) \Rightarrow (2)$ and $(3) \Rightarrow (1)$.

  \smallskip
  \noindent $(1) \Rightarrow (2)$: Let $L$ be a regular language with
  neutral element and $A$ be the minimal DFA accepting $L$. Then, the
  \dynprop program $P$, accepting $\dynamic{L}$, constructed in the
  proof of Proposition~\ref{prop:regprop} accepts exactly
  $\dynamicAlt{L}$.

  It should be clear that the correctness of the update functions of
  $P$ carries over immediately to the current setting. To see that
  also the initialization of the different relations is correct,
  notice that, as $A$ is minimal and $a$ is a neutral element, it
  must hold that $(q,a,p) \in \delta$ iff $q = p$, for all states $p$
  and $q$ of $A$. Since $\word(E^a_n) = a^n$ it now follows from this
  observation that the different relations are properly initialized.

  \smallskip
  \noindent $(3) \Rightarrow (1)$: Let $L$ be a language such that
  $\dynamicAlt{L}$ is accepted by a $\dynprop(\Rel,\Rel)$ program
  $P$. We show that $L$ is regular by constructing a finite automaton
  accepting $L$. Again, this can be done almost identically as in the
  proof of implication $(3) \Rightarrow (1)$ in
  Theorem~\ref{theo:dynprop-reg}. The key point to notice is that a
  position which is labeled $a$ in the current semantics can
  intuitively be seen as an empty, i.e. not-labeled, position in the
  original semantics due to the fact that $a$ is a neutral
  element.

  Therefore, we proceed in two steps. First, completely ignoring the
  symbol $a$, we create the automaton $A$ exactly as in the proof of
  Theorem~\ref{theo:dynprop-reg}. Denote $\Sigma \setminus \{a\}$ by
  $\Sigma_a$. Then, as before, it can be shown that $L(A) = L
  \cap \Sigma_a^*$, i.e. $A$ accepts all strings in $L$ that do not contain an
  $a$. Now, as $a$ is a neutral element of $L$, it holds that $L =
  \bigcup_{w = \sigma_1\cdots\sigma_n \in L(A)} L(a^*\sigma_1a^*\cdots
  a^*\sigma_na^*)$. Hence, the desired automaton $A'$, with $L(A') =
  L$ can be obtained from $A$ by adding the transition $(q,a,q)$ to
  $A$, for all states $q$ of $A$.
\qed
\bigskip
\noindent
{\it Regular Tree Languages.}
We now investigate the dynamic complexity of the regular tree
languages. Thereto, we first define dynamic
tree language. A tree $t$ over an alphabet $\Sigma$ is
encoded by a logical structure $T$ with as universe the first $n$ elements of the list
  $(1,11,12,111,112,\allowbreak 121,122,\ldots)$, for some $n \in \nat$, and
  consisting of (1) one unary relation $R_\sigma$, for each symbol
  $\sigma\in\Sigma$, (2) a constant $\rootnode$, denoting the element
  $1$, and (3) binary relations $\lchild$ and $\rchild$, containing all tuples
  $(u,u1)$ and $(u,u2)$, respectively.

  The updates are terms $\ins_\sigma(u)$ and $\reset(u)$, setting and
  resetting the label of node $u$ in $T$, exactly as in the string
  case. So, the logical structure $T$ is a fixed balanced binary tree
  in which the labels can change. Then, the tree $t$ encoded by $T$ is
  the largest subtree of $T$ whose root is the element $1$ and in
  which all nodes are labelled with an alphabet symbol. Notice that a
  node of $T$ is included in $t$ if it, and all its ancestors, carry
  an alphabet symbol.

Exactly as for the word languages, for a tree language $L$, we let
$\dynamic{L}$ be the set of update sequences leading to a tree $t \in
L$. A dynamic program works on a dynamic tree language exactly as it
does on a dynamic language. We then obtain the following result.

\bigskip

\begin{proposition}\label{prop:trees}
  Let $L$ be a regular tree language. Then, $\dynamic{L} \in
  \dynprop(\Fun,\Rel)$.
\end{proposition}

\proof
  We first introduce some notation. For a node $u$ of $T$, let $\subtree^u_T$ be
the largest subtree of $T$ whose root is $u$ and in which all nodes
are labelled with an alphabet symbol. Hence, $T$ encodes the tree
$\subtree^{\rootnode}_T$. Further, for a tree $t$, we denote its set
  of nodes by $\nodes(t)$, and for $u \in \nodes(t)$, $\lab_t(u)$
  denotes the \emph{label} of $u$ in $t$. 
\medskip

\noindent  The program will make use of the following precomputed relations and
  functions on $T$:
  \begin{itemize}
  \item a binary relation $\anc$, such that $\anc(x,y)$ holds if $x$
    is an ancestor of $y$;
  \item a binary funtion $\lca$, such that $\lca(x,y) = z$ if $z$ is
    the least common ancestor of $x$ and $y$;
  \item a unary function $\parent$ such that $\parent(u) = v$ if
    $\lchild(v,u)$ or $\rchild(v,u)$, and $\parent(u) = u$, if $u = \rootnode$;
  \item unary functions $\lchildfun$ and $\rchildfun$ such that $\lchildfun(u)
    = v$ if $\lchild(u,v)$ and $\lchildfun(u) = u$, otherwise; and $\rchildfun(u)
    = v$ if $\rchild(u,v)$ and $\rchildfun(u) = u$, otherwise.
  \end{itemize}

  Let $L$ be a regular (binary) tree language, and $A =
  (Q,\delta,(q_\sigma^I)_{\sigma \in \Sigma},F)$ be a bottom-up
  deterministic tree automaton accepting $L$, with $\delta: Q \times Q
  \times \Sigma \to Q$ the (complete) transition function. A
  \emph{run} of a $A$ on a tree $t$ is a mapping $\rho: \nodes(t)
  \rightarrow Q$ such that (1) for all leaf nodes $u$ of $t$, $\rho(u)
  = q_\sigma^I$, where $\lab_t(u) = \sigma$, and (2) for all
  non-leaf nodes $u$, with children $u_1,u_2$, we have
  $\delta(\rho(u_1),\rho(u_2),\lab(u)) = \rho(u)$. If $\rho(\rootnode)
  = q$, we say that $\rho$ is a run of $A$ on $t$ \emph{to} $q$. A tree
  $t$ is accepted if there is a run of $A$ on $t$ to $q_f$, for some
  $q_f \in F$.
\smallskip

  We denote by $\subtree^{u,v}_t$ the subtree of $t$ with root $u$
  which contains all descendants of $u$ but no descendants of
  $v$. For such a tree $\subtree^{u,v}$ we will also be interested in
  runs which assign a state $p$ to the new leaf node $v$, not
  necessarily consistent with the label of $v$, and are valid runs
  otherwise. Thereto, a function $\rho: \nodes(\subtree^{u,v}_t)
  \rightarrow Q$ is a run of $A$ on $\subtree^{u,v}_t[v \rightarrow
  p]$ to $q$ iff $\rho(u) = q$, $\rho(v) = p$, and $\rho$ is a valid
  run of $A$ on $\subtree^{u,v}_t$, except for the fact that $p =
  q^I_{\lab(v)}$ does not have to hold.

Before giving the relations we will maintain, we define a few
subformulas which will be used several times in the subsequent
definitions and formulas. 
$$\ancself(x,y) \equiv \anc(x,y) \vee x = y,$$
$$\eps(x) \equiv \bigwedge_{\sigma \in \Sigma} \neg R_\sigma(x),
\mbox{ and }$$
$$\leaf(x) \equiv (\lchildfun(x) = x \vee \rchildfun(x) = x \vee
(\eps(\lchildfun(x)) \wedge \eps(\rchildfun(x)))$$
\newpage
  Our dynamic program will maintain the following relations:
  \begin{itemize}
  \item $\con = \{(x,y) \mid \ancself(x,y) \wedge \forall z \mbox{ with }
\ancself(x,z) \wedge \anc(z,y) \mbox{, $R_\sigma(z)$ is true,}$
  for some
$\sigma \in \Sigma\}$
\item $R_{q} = \{x \mid \mbox{ there is a run of $A$ on $\subtree^x$ to
  $q$}\}, \mbox{ and }$
\item $R_{q_1,q_2} = \{(x_1,x_2) \mid \mbox{ there is a run of $A$ on
  $\subtree^{x_1,x_2}[x_2 \rightarrow q_2]$ to $q_1$}\}$
  \end{itemize}

That is, the relation $\con$ expresses whether elements $x$ and $y$
are \emph{connected} in $T$, i.e. whether all nodes on the path from
$x$ to $y$, except possibly $y$ itself, carry an alphabet symbol. The
relation $R_q$ contains all nodes $x$ for which there is a run on
$\subtree^x$ to $q$, and $(x_1,x_2) \in R_{q_1,q_2}$ intuitively holds
if, assuming there is a run on $\subtree^{x_2}$ to $q_2$, then there
is a run on $\subtree^{x_1}$ to $q_1$. 

First of all, due to Lemma~\ref{lem:init} we can assume that these
relations are initialized correctly as follows:
\begin{itemize}
\item $\con = \{(x,x)\},$
\item for all $q \in Q$, $R_{q} = \emptyset$ , and
\item for all $q_1,q_2 \in Q$,  $R_{q_1,q_2} = \emptyset$ if $q_1 \neq q_2$, and  $R_{q_1,q_2} = \{(x,x)\}$, otherwise.
\end{itemize}

We now give the update formulae for the different relations. First,
the relation $\con$ can easily be maintained. For all $\sigma \in \Sigma$,
\begin{multline*}
\phi^\con_{\ins_\sigma}(y;x_1,x_2) \equiv \big[\neg (\ancself(x_1,y)
\wedge \anc(y,x_2)) \wedge \con(x_1,x_2)\big] \vee \\
\big[\con(x_1,y) \wedge (\con(\lchildfun(y),x_2) \vee \con(\rchildfun(y),x_2))\big]
\end{multline*}

$$
\phi^\con_{\reset}(y;x_1,x_2) \equiv \neg (\ancself(x_1,y)
\wedge \anc(y,x_2)) \wedge \con(x_1,x_2)
$$

Before giving the update formulae for $R_q$ and $R_{q_1,q_2}$ we
define a formula which will be used several times. For $p \in Q$ and
$\sigma \in \Sigma$, the following formula intuitively says ``if node
$x$ is labeled $\sigma$, then there is a run on $\subtree^x$ to $p$'':
\[
\phi_\sigma^p(x) \equiv \big[ (\leaf(x) \wedge q^I_\sigma =
p \big] \vee 
\big[\neg \leaf(x) \wedge \bigvee_{\substack{p_1,p_2 \in
    Q\\\delta(p_1,p_2,\sigma) = p}} (
R_{p_1}(\lchildfun(x)) \wedge R_{p_2}(\rchildfun(x))) \big]
\]

We can now give the different update formula for the insert operation. For
each $\sigma \in \Sigma$ and $q \in Q$, the relation $R_q$ can be
updated as follows
\begin{multline*}
  \phi_{\ins_\sigma}^{R_q}(y;x) \equiv \big[\neg (\ancself(x,y) \wedge
  \con(x,y)) \wedge R_{q}(x)\big] \vee \\ \big[\ancself(x,y)
  \wedge \con(x,y) \wedge 
\bigvee_{p \in Q} (\phi^p_\sigma(y) \wedge R_{q,p}(x,y)\big]
\end{multline*}

The update formula for ${R_{q_1,q_2}}$ is similar but more
involved. It is defined as follows
$$\phi_{\ins_\sigma}^{R_{q_1,q_2}}(y;x_1,x_2) \equiv
\ancself(x_1,x_2) \wedge \phi^\con_{\ins_\sigma}(y;x_1,x_2) \wedge (\phi_1 \wedge \phi_2
\wedge \phi_3 \wedge \phi_4 \wedge
\phi_5),$$ where $\phi_1$ to $\phi_5$ are formulas defined according
to the position of $y$ with respect to $x_1$ and $x_2$:
\begin{itemize}
\item If $y$ does not occur in $\subtree^{x_1,x_2}$ after
  $\ins_\sigma(y)$, or $y = x_2$, then the truth value of
  $R_{q_1,q_2}(x_1,x_2)$ is not changed: $$\phi_1 \equiv (\neg
  \con(x_1,y) \vee \ancself(x_2,y)) \wedge R_{q_1,q_2}(x_1,x_2)$$
\item Let $\lca(x_2,y) = z$. If $y = z$, and $x_2$ is a left
  descendant of $y$, i.e. $\ancself(\lchildfun(y),x_2)$, we can determine the state $p$ of $z$ and use
  this information to decide whether $R_{q_1,q_2}(x_1,x_2)$:
\begin{multline*}
\phi_2 \equiv y = z \wedge \ancself(\lchildfun(y),x_2) \wedge \\ \bigvee_{\substack{p,p_1,p_2 \in
    Q\\\delta(p_1,p_2,\sigma) = p}} (R_{p_1,q_2}(\lchildfun(y),x_2)
\wedge R_{p_2}(\rchildfun(y)) \wedge R_{q_1,p}(x_1,y))
\end{multline*}

\item Else if $y = z$, and $x_2$ is a right descendant of $y$, then
  $\phi_3$ is almost identical to $\phi_2$.
\item Else if $y \neq z$ and $y$ is a left descendant of $z$, then: 
  \begin{multline*}
    \phi_4 \equiv y \neq z \wedge \ancself(\lchildfun(z),y) \wedge
    \bigvee_{\substack{p \in Q}} \Big( \phi_\sigma^p(y) \wedge \\
    \bigvee_{\substack{r,r_1,r_2 \in Q, \sigma' \in
        \Sigma\\\delta(r_1,r_2,\sigma') = r}}\big[ R_{\sigma'}(r)
    \wedge R_{r_1,p}(\lchildfun(z),y) \wedge
    R_{r_2,q_2}(\rchildfun(z),x_2) \wedge R_{q_1,r}(x_1,z) \big] \Big)
  \end{multline*}
\item Else if $y \neq z$ and $y$ is a right descendant of $z$, then
  $\phi_5$ is almost identical to $\phi_4$:
\end{itemize}

We now give the different formulae for the reset operation. Again, we
first define a subformula which will be used several times. The
following formula intuitively says ``if node $y$ is reset, and $y'$ is
its parent, then there is a run on $\subtree^{y'}$ to $p$'':
\begin{multline*}
  \psi^p(y,y') \equiv \bigvee_{\substack{\sigma \in \Sigma\\p =
      q^I_\sigma}} R_\sigma(y')  \wedge \\ \big[(\lchildfun(y') = y \wedge \eps(\rchildfun(y')))
\vee (\rchildfun(y') = y \wedge \eps(\rchildfun(y')))\big]
\end{multline*}

We can now define the different formulae for the reset operation. For
all $q \in Q$, 
\begin{multline*}
\phi_{\reset}^{R_q}(y;x) \equiv \big[\neg(\ancself(x,y) \wedge 
\con(x,y)) \wedge R_q(x)\big] \vee\\ \big[\anc(x,y) \wedge \bigvee_{p
  \in Q}\psi^{p}(y,\parent(y)) \wedge R_{q,p}(x,\parent(y))\big]  
\end{multline*}

Again, the formula $\phi_{\reset}^{R_{q_1,q_2}}$ is similar but more involved
$$\phi_{\reset}^{R_{q_1,q_2}}(y;x_1,x_2) \equiv \ancself(x_1,x_2)
\wedge \neg(\ancself(x_1,y) \wedge \anc(y,x_2)) \wedge\con(x_1,x_2)
\wedge (\phi_1 \vee \phi_2 \vee \phi_3)$$

Notice that if any of these conditions is not satisfied then
$R_{q_1,q_2}(x_1,x_2)$ cannot hold after $\reset(y)$. The formulas
$\phi_1$, $\phi_2$ and $\phi_3$ depend on the possible remaining
positions of $y$ w.r.t. $x_1$ and $x_2$. We only have to distinguish
three cases here, opposed to five before, because we do not have to
consider the case $\lca(y,x_2) = y$ anymore. Indeed, if $\lca(y,x_2) =
y$, resetting $y$ immediately disconnects $x_1$ from $x_2$.

\begin{itemize}
\item If $y$ does not occur in $\subtree^{x_1,x_2}$, then the truth
  value of $R_{q_1,q_2}(x_1,x_2)$ is not changed: $$\phi_1 \equiv
  (\neg \con(x_1,y) \vee \ancself(x_2,y)) \wedge
  R_{q_1,q_2}(x_1,x_2)$$
\item Let $\lca(x_2,y) = z$ and $\parent(y) = y'$. If $y \neg z$ and
  $y$ is a left descendant of $z$, then:
  \begin{multline*}
    \phi_4 \equiv y \neq z \wedge \ancself(\lchildfun(z),y) \wedge
    \bigvee_{\substack{p \in Q}} \Big( \psi^p(y,y') \wedge \\
    \bigvee_{\substack{r,r_1,r_2 \in Q, \sigma' \in
        \Sigma\\\delta(r_1,r_2,\sigma') = r}}\big[ R_{\sigma'}(z)
    \wedge R_{r_1,p}(\lchildfun(z),y') \wedge
    R_{r_2,q_2}(\rchildfun(z),x_2) \wedge R_{q_1,r}(x_1,z) \big] \Big)
  \end{multline*}
\item If $y \neq z$ and $y$ is a right descendant of $z$, the formula
  $\phi_3$ is almost identical to $\phi_2$.
\end{itemize}

Finally, the update formulae for the acceptence relation depend only
on the new value of the relations $R_q$, for $q \in Q$. That is, for
all $\sigma \in \Sigma$,
$$\phi_{\ins_\sigma}^{\accept}(x) = \bigvee_{q \in
  F}\phi_{\ins_\sigma}^{R_q}(x,\rootnode)$$ and
$$\phi_{\reset}^{\accept}(x) = \bigvee_{q \in
  F}\phi_{\reset}^{R_q}(x,\rootnode)$$

\qed

\section{Beyond Formal Languages}\label{sec:graphs}

The definitions given in Section~\ref{sec:defs} only concerned dynamic
problems for word structures. Following~\cite{WeberS05}, we now extend
these definitions to arbitrary structures. Thereto, let $\voc$ be a
vocabulary containing relation symbols of arbitrary arities. 
We assume that a structure over $\voc$ of size $n$ has as universe
$\{1,\ldots,n\}$.
The empty structure over vocabulary $\gamma$ of size $n$ and only
empty relations is denoted $E_n(\voc)$.

The set of \emph{abstract updates} $\Delta(\voc)$ is defined as
$\{\ins_R, \del_R \mid R \in \voc\}$. A \emph{concrete update} is a
term of the form $\ins_R(i_1,\ldots,i_k)$ or $\del_R(i_1,\ldots,i_k)$,
where $k = \arity(R)$. 
A concrete update is \emph{applicable} in a structure of size $n$ if
$i_j \le n$, for all $j \in [1,k]$.  By $\Delta_n(\voc)$ we denote the
set of applicable concrete updates for structures over $\voc$ of size
$n$. For a sequence $\alpha=\delta_1\ldots
\delta_k\in(\Delta_n(\voc))^+$ of updates we define $\alpha(A)$ as
$\delta_k(\ldots(\delta_1(A))\ldots)$,  
where $\delta(A)$ is the structure obtained from $A$ by setting $R(i_1,\ldots,i_k)$ to true if $\delta =
  \ins_R(i_1,\ldots,i_k)$; and setting $R(i_1,\ldots,i_k)$ to false if $\delta =
  \del_R(i_1,\ldots,i_k)$.

\begin{definition}
  Let $\voc$ be a vocabulary, and $F$ be a set of $\voc$-structures. The
  \emph{dynamic problem} \dynamic{F} is the set of all pairs
  $(n,\alpha)$, with $n > 0$ and $\alpha \in (\Delta_n(\voc))^+$ such
  that $\alpha(E_n(\voc)) \in F$. We call $F$ the \emph{underlying static problem} of
$\dynamic{F}$.
\end{definition}

We now explain how a dynamic program operates
on a dynamic problem. For a program $P$, there again is a program
state $S$ containing the current structure and auxiliary relations,
one of which is $\accept$, which are updated according to the updates
which occur and the update functions of $P$. The state $S$ is
\emph{accepting} iff $S \models \accept$. Then, let $L(P) =
\{(n,\alpha) \mid \alpha \in (\Delta_n(\voc))^* \mbox{ and }
\alpha(E_n'(\voc)) \mbox{ is }\allowbreak \mbox{accepting}\}$, where $E_n(\voc)'$ denotes
the structure $E_n(\voc)$ extended with empty auxiliary relations.

A program $P$ \emph{accepts} a problem $F$ iff $L(P) = \dynamic{F}$. If $P
\in \mcC$, for some dynamic complexity class $\mcC$, we also write
$\dynamic{F} \in \mcC$.
\medskip

\noindent
{\it Incomparability of FO and \dynprop.} As we have seen in the previous sections, when restricted to
monadic input schemas, \dynprop in a sense has the power of
MSO. However, if we add one binary relation \dynprop cannot even
capture first-order logic. This is also true if we allow the program to use precomputed functions from the set $\SetSucc$.

Thereto we will consider \emph{alternating graphs}, coded via the binary edge relation $E$ and two unary relations $A$ and $B$ that form a decomposition of the universe $V$ into the set of \emph{existential} and \emph{universal} nodes. Given a node $s\in V$, the set of all \emph{reachable} nodes $\reachable(s)$ is defined as the smallest set satisfying 
{
\begin{itemize}
	\item $s\in\reachable(s)$,
	\item if $u\in A$ and there is a $v\in\reachable(s)$ such that $(u,v)\in E$, then $u\in\reachable(s)$,
	\item if $u\in B$ and for all $v\in V$ with $(u,v)\in E$ we have $v\in\reachable(s)$, then $u\in\reachable(s)$.
\end{itemize}
}
Now we define \altreach as the problem, given an alternating graph
$G=(A\dot\cup B,E)$ and two nodes $s$ and $t$, is
$t\in\reachable(s)$. We note that \altreach is P-complete (see for
example \cite{vollmerCircuitComplexity}).  

\begin{proposition}\label{prop:altreachDynPropRelRel}
  $\dynamic{\altreach} \notin \dynprop(\Rel,\Rel)$
\end{proposition}

Before we can prove the proposition we will state a lemma that describes
an important property of \dynprop-algorithms. An update sequence
\emph{working on $k$-tuples} $\alpha$ is a sequence of updates over
the (abstract) universe $\{1,...,k\}$. Given a $k$-tuple $\bar
i=(i_1,...,i_k)$, $\alpha(\bar i)$ will denote the sequence of updates
one obtains when applying the updates on the elements of the $k$-tuple
$\bar i$, so instead of using the (abstract) universe element $x$ the
element $i_x$ should be used. For example the update $\ins_R(1,4,2)$
would result in an update $\ins_R(i_1, i_4, i_2)$. We also use the
notion of types as introduced in the proof of
Proposition~\ref{prop:propreg}.

\begin{mylemma}\label{lemma:SequenceWorkingOverTuples}
  Let $\alpha$ be a sequence of updates working on $k$-tuples. Let $P$
  be a $\dynprop(\Rel,\Rel)$ program, $S$ a state of $P$ and consider
  two tuples of elements $\bar i=(i_1, ..., i_k)$ and $\bar j=(j_1,
  ..., j_k)$ of $S$ such that $\type{S}{\bar i} = \type{S}{\bar
    j}$. Then, $\type{\alpha(\bar i)(S)}{\bar i} = \type{\alpha(\bar
    j)(S)}{\bar j}$, i.e. the type of $\bar i$ after applying
  $\alpha(\bar i)$ and the type of $\bar j$ after applying
  $\alpha(\bar j)$ are still the same. In particular, the value of the
  $\accept$-relation is the same in $\alpha(\bar i)(S)$ and
  $\alpha(\bar j)(S)$.
\end{mylemma}
\proof It suffices to consider one update operation $\delta$ working
on $k$-tuples. Then the lemma follows by induction on the length of
the update sequence $\alpha$. Let $\iota$ be the tuple of elements in
$\{1,\dots,k\}$ being the parameters of $\delta$. And consider any
(auxiliary) relation $R$ updated by the program $P$ on a tuple
$\kappa$ also of elements in $\{1,\dots,k\}$. Let $\iota(\bar i)$,
$\iota(\bar j)$, $\kappa(\bar i)$ and $\kappa(\bar j)$ denote the
corresponding tuples in state $S$. Then the evaluation of the update
formula for $R$ on $\kappa(\bar i)$ after the operation $\delta$ with
parameters $\iota(\bar i)$ is dependend only on the type of the
elements in $\kappa(\bar i)\cup\iota(\bar i)$. The same holds for the tuples
corresponding to $\bar j$. Since the types $\type{S}{\bar i}$ and
$\type{S}{\bar j}$ coincide, the update formula evaluates to the same
value.  \qed

\noindent{\it Proof of Proposition~\ref{prop:altreachDynPropRelRel}.}
  We first define a family of alternating graphs $\mathcal{G} = \{G_m
  \mid m\in\mathbb{N}\}$. Every graph $G_m$ consists of the following
  set $V_m$ of nodes:
\begin{itemize}
	\item two nodes $s$ and $t$
	\item a set of $2m$ nodes $P=\{p_1,...,p_{2m}\}$
	\item for each subset $I$ of $P$ of size $m$ a node $q_I$, forming the set $Q$ (of size ${2m \choose m}$)
	\item for each subset $J$ of $Q$ a node $r_J$, forming the set $R$ (of size $2^{|Q|}$).
\end{itemize}
All nodes are existential nodes except the nodes in set $Q$, which are
universal. Further, the following set of edges $E_m$ is already
present in the graph $G_m$:
\begin{itemize}
	\item for each subset $I$ of $P$ of size $m$ the set of edges $\{(q_I, p) \mid p\in I\}$ and
	\item for each subset $J$ of $Q$ the set of edges $\{(r_J, q) \mid q\in J\}$.
\end{itemize}
As updates we will only consider insertions of edges from $s$ to nodes in the set $R$ and from nodes in the set $P$ to $t$.  

We will show that no dynamic program can maintain auxiliary relations
such that it can incrementally answer the question whether $t$ is
reachable from $s$ in the alternating graph, starting from any $G_m$
and arbitrary precomputation on the auxiliary relations. This will
prove the claimed proposition.

We will make use of the following two lemmas:
\begin{mylemma}\label{lemma:reachable}
  For every $m$ and every pair of distinct nodes $r, r'\in R$ of $G_m$ there
  exists a set $I\subset P$ of size $m$ such that in the graph $G'_m
  := (V_m, E_m \cup \bigcup_{p\in I} (p,t))$ it holds that
  $t\in\reachable(r)$ and $t\notin\reachable(r')$.
\end{mylemma}
\begin{proof}
  Each of the nodes in $R$ corresponds to some (different) subset of
  $Q$. Hence, by definition of $G_m$, there must exist a node $q_I\in Q$
  which in $G_m$ is reachable from $r$ but not from $r'$. We show that
  the set $I \subset P$ is the desired set,
  i.e. for $G'_m := (V_m, E_m \cup \bigcup_{p\in I} (p,t))$ it holds
  that $t\in\reachable(r)$ and $t\notin\reachable(r')$. Thereto,
  notice that in $G'_m$ it holds that $q_I$ is the only node in the set
  $Q$ such that $t\in\reachable(q_I)$. This follows from the fact that
  all nodes in $Q$ are universal nodes. But now, as $r$ and $r'$ are
  existential nodes, it holds that $t\in\reachable(r)$ and
  $t\notin\reachable(r')$, which concludes the proof.
\end{proof}
\begin{mylemma}\label{lemma:ktype}
The number of possible $k$-types of a structure with $x$ auxiliary relations of maximal arity $y$ is bounded by $2^{x\cdot k^y}$.
\end{mylemma}
\begin{proof}
  A $k$-type is constructed from a set of atoms $R(\bar j)$ (where
  each element in $\bar j$ is in $[1,k]$) by adding either $R(\bar j)$
  or $\neg R(\bar j)$ to the $k$-type. Hence, there exist at most
  $2^{|\text{atoms}|}$ different $k$-types where $|\text{atoms}|$
  denotes the number of different atoms. For one $y$-ary relation $R$
  all atoms of $R$ can be seen as the set of all $y$-tuples of
  elements in $\bar i$. So, one relation of arity $y$ produces $k^y$
  different atoms. As there are $x$ different relations, there are
  hence at most $x\cdot k^y$ atoms, and thus at most $2^{x\cdot k^y}$
  different $k$-types.
\end{proof}
\noindent Now, assume, towards a contradiction, there exists a dynamic program $P$ for
$\dynamic{\altreach}$ in $\dynprop(\Rel,\Rel)$ that makes use of $a$
auxiliary relations of maximal arity $b$. For a graph $G_m$ and all
nodes $r\in R$, we will consider the tuples $V_r :=
(s,t,r,p_1,...,p_{2m})$.

Since 
\[
|Q| = {2m \choose m} = \prod_{i=0}^{m-1} \frac{2m-i}{m-i} \geq \prod_{i=0}^{m-1} \frac{2m}{m} = 2^m \textrm{ and so } |R|\geq 2^{2^m} 
\]
there exists a number $m$ such that $|R|$ is bigger than the number of
$(2m+3)$-types in any state $S$ of $P$. Indeed, from
Lemma~\ref{lemma:ktype} and the fact that the program can use $a+6$
relations (the auxiliary relations, the input relations $E$, $A$, and
$B$, and the equality, order and \accept relations) of maximal arity
$b$, it follows that the number of $(2m+3)$-types in $S$ is bounded by
$2^{(a+6)\cdot(2m+3)^b}$. For a large enough value of $m$, this is
clearly dominated by $2^{2^m}$. Now, fix such a large enough $m$ and
corresponding graph $G_m$, and let $S$ be a state $P$ is in when the
current graph is $G_m$. Then, due to the above reasoning, in the set
$R$ (of $G_m$) there must exist two distinct elements $r$ and $r'$
such that $\type{S}{V_r} = \type{S}{V_{r'}}$.

Now, according to Lemma~\ref{lemma:reachable}, we can find a set $I$
of $m$ elements in $P$ such that after the insertion of all edges
$\{(p,t)\mid p\in I\}$ it holds that $t\in\reachable(r)$ and
$t\notin\reachable(r')$. Let $I = \{p_{i_1}, ..., p_{i_m}\}$ and consider
the two sequences of update operations
\begin{eqnarray*}
\alpha  & = &  (\ins_E(p_{i_1},t), ..., \ins_E(p_{i_m}, t), \ins_E(s,r))  \\   
\alpha' & = &  (\ins_E(p_{i_1},t), ..., \ins_E(p_{i_m}, t), \ins_E(s,r')).    
\end{eqnarray*}
Notice that $\alpha(G_m)$ yields a graph in which $t \in
\reachable(s)$, whereas $t \notin \reachable(s)$ in
$\alpha'(G_m)$. However, as $\type{S}{V_r} = \type{S}{V_{r'}}$, it
follows from Lemma~\ref{lemma:SequenceWorkingOverTuples} that also
$\type{\alpha(S)}{V_r} = \type{\alpha'(S)}{V_{r'}}$. Hence, $P$ will
either in both cases claim that $t \in \reachable(s)$ (if $\accept$
holds in $\type{\alpha(S)}{V_r}$) or claim in both cases that $t
\notin \reachable(s)$. We can conclude that there does not exist a
$\dynprop(\Rel,\Rel)$ program for $\altreach$.
\qed

This proof can be adapted to show that even with a precomputed successor-relation, one cannot maintain the reachability problem in alternating graphs. 

\begin{proposition}\label{proposition:AltReachDynSucc}
  $\dynamic{\altreach} \notin \dynprop(\SetSucc,\Rel)$
\end{proposition}
  In order to prove this we need an observation similar to
  Lemma~\ref{lemma:SequenceWorkingOverTuples} for $\dynprop(\SetSucc,
  \Rel)$. For an element $i$ and a number $l$ let the
  \emph{$l$-neighborhood} of $i$, denoted $\neighborhood{l}(i)$, be
  the following tuple of elements:
\[
\left(\preM^l(i), \preM^{l-1}(i), ..., \preM(i), i, \succM(i), ..., \succM^{l-1}{i}, \succM^l(i)\right).
\]
For a tuple of elements $\bar i$, we denote by $\neighborhood{l}(\bar i)$
the tuple $\left(\neighborhood{l}(\minM), \neighborhood{l}(i_1), ...,\neighborhood{l}(i_k)\right)$. 

\begin{mylemma}\label{lemma:SequenceWorkingOverTuplesSucc}
  Let $\alpha$ be a sequence of updates working on $k$-tuples such
  that $|\alpha|=l$. For each $\dynprop(\SetSucc,\Rel)$ program $P$,
  there exists a number $c$, depending only on $P$, such that the
  following holds: let $S$ be some state of $P$ and consider
  two tuples of elements $\bar i=(i_1, ..., i_k)$ and $\bar j=(j_1,
  ..., j_k)$ of elements of $S$ such that
  $\type{S}{\neighborhood{c\cdot l}(\bar i)} =
  \type{S}{\neighborhood{c\cdot l}(\bar j)}$. Then, $\type{\alpha(\bar
    i)(S)}{\bar i} = \type{\alpha(\bar j)(S)}{\bar j}$.
\end{mylemma}
\proof The $c$ of the lemma is the maximal nesting dephth of the used
functions \succM and \preM (for example the term
$\succM(\succM(\preM(x)))$ has nesting depth 3) in $P$. Let
$\alpha_n$ be the prefix of length $n$ of the update sequence
$\alpha$. We will here prove the slightly stronger statement that,
assuming the conditions of the lemma, $\type{\alpha_n(\bar
  i)(S)}{\neighborhood{c\cdot(l-n)}(\bar i)} = \type{\alpha_n(\bar
  j)(S)}{\neighborhood{c\cdot(l-n)}(\bar j)}$. Then the lemma follows
because $\neighborhood{0}(\bar i) = (\minM,\bar i)$. The proof works
by induction on $n$ (assuming $n<l$). For $n=0$ the statement is
contained in the condition of the lemma. So assume the statement holds
for $n<l$, we show that it still holds for $n+1$. Let $\delta$ be the
update such that $\alpha_{n+1} = \alpha_n\delta$. Just as in the proof
of Lemma~\ref{lemma:SequenceWorkingOverTuples} consider any
(auxiliary) relation $R$ updated by the program $P$ on elements in
$\neighborhood{c\cdot(l-(n+1))}(\bar i)$. The evaluation of the update
formula for $R$ is dependend only on the type of
$\neighborhood{c\cdot(l-n)}(\bar i)$. This is true because one can
reach other elements in the universe only by using the
functions. Since these are nested at most $c$ times, from any element
in $\neighborhood{c\cdot(l-(n+1))}(\bar i)$ only elements in
$\neighborhood{c\cdot(l-n)}(\bar i)$ can be reached. The same holds
for the tuples corresponding to $\bar j$. As $\type{\alpha_n(\bar
  i)(S)}{\neighborhood{c\cdot(l-n)}(\bar i)} = \type{\alpha_n(\bar
  j)(S)}{\neighborhood{c\cdot(l-n)}(\bar j)}$, and $\alpha_{n+1} =
\alpha_{n+1}\delta$, it follows that $\type{\alpha_{n+1}(\bar
  i)(S)}{\neighborhood{c\cdot(l-(n+1))}(\bar i)} = \type{\alpha_{n+1}(\bar
  j)(S)}{\neighborhood{c\cdot(l-(n+1))}(\bar j)}$.  \qed 

\proof[Proof
of Proposition~\ref{proposition:AltReachDynSucc}.]  The proof now
follows the lines of the proof of
Theorem~\ref{prop:altreachDynPropRelRel}. Assume that there exists a
$\dynprop(\SetSucc, \Rel)$ program $P$ for $\dynamic{\altreach}$
making use of $a$ auxiliary relations of maximal arity $b$. We will
again consider (in a graph $G_m$) the tuples $V_r$ and $V_{r'}$, but now 
their corresponding $(m+1)c$-neighborhoods $\neighborhood{(m+1)c}(V_r)$ and
$\neighborhood{(m+1)c}(V_{r'})$, where $c$ is the constant only depending
on $P$ of Lemma~\ref{lemma:SequenceWorkingOverTuplesSucc}. Using Lemma~\ref{lemma:ktype} we know that
in any state $S$ of $P$ the number of types of these neighborhoods is
bounded by $2^{(a+6)\cdot((2(m+1)c+1)(2m+3+1))^b}$. Hence we can again find
a number $m$ big enough such that there are distinct $r,r' \in R$ in $G_m$ such
that $\type{S}{\neighborhood{(m+1)c}(V_r)}$ and
$\type{S}{\neighborhood{(m+1)c}(V_{r'})}$. Using the same argument as above
and Lemma~\ref{lemma:SequenceWorkingOverTuplesSucc} we then get the
desired contradiction.  \qed

\begin{remark}
The proofs of the foregoing lemma and proposition depends heavily on the fact that the neighborhood of a tuple increases with each update operation only by a constant additional term. This is because the two used functions $\preM$ and $\succM$ are complementary in that $\preM(\succM)=\succM(\preM)$. So the order of their usage is not important. If one allows two independend functions (for example two different successor-functions on the universe) the size of the neighborhood possibly doubles after each operation so the proof of the proposition (based on a counting argument) would not work.
\end{remark}

In fact from the proof of the above proposition one can conclude an even stronger
statement. The graphs used in the proof are very restricted in the sense that the
length of the longest path is bounded by a constant. Let
\bdaltreach{d} be the alternating reachability problem on graphs of
depth at most $d$. It is easily seen that \bdaltreach{d} is
expressible by a FO-formula, so we get the following
\begin{theorem}\label{theo:fo-neq-dynprop}
  There exists a problem $F\in\text{FO}$ such that $\dynamic{F} \notin \dynprop(\SetSucc, \Rel)$.\qed
\end{theorem}
On the other hand the reachability problem on acyclic deterministic directed graphs can be maintained in $\dynprop$ (Hesse \cite{Hesse03a}) but cannot be expressed via an FO-formula (as can be easily seen by standard EF-games arguments). So these classes are incomparable.
\smallskip

\noindent
{\it Using functions to maintain EFO.} Next we exhibit a
class of properties which can be maintained in $\dynqf$ with
precomputation. An \emph{existential first-order (EFO)} sentence is
a first-order sentence of the form $\exists x_1, \ldots x_k
\phi(\bar{x})$, where $\phi(\bar{x})$ is a quantifier free
formula. 
\begin{theorem}\label{theorem:EFOdynprop}
  For any EFO-definable problem $F$, $\dynamic{F} \in
  \dynprop(\Fun,\Fun)$
\end{theorem}

\proof 
  Let $\psi = \exists x_1, \ldots x_k \phi(\bar{x})$ be an
  EFO-sentence over vocabulary $\gamma$. We show that there exists a
  $\dynprop(\Fun,\Fun)$ program $P$ which maintains whether $A \models
  \psi$, for any $\gamma$-structure $A$.

  We first introduce some notation. A tuple $\bar i=(i_1,\ldots,i_l)$
  is \emph{disjoint} if $i_j \neq i_k$, for all $j,k \in [1,\ell]$,
  with $j \neq k$. A \emph{disjoint type} is the type of a disjoint
  tuple. For a type $\tau$, let $\phi_\tau$ be an EFO sentence which
  is satisfied in a structure $A$ iff $A$ contains a tuple $\bar{x}$
  such that $\type{A}{\bar{x}} = \tau$.

  Now, it is well known and easy to see that for any EFO sentence
  $\phi = \exists x_1, \ldots x_k \psi(\bar{x})$ there exists a set
  $\theta_\psi$ of disjoint $\ell$-types, with $\ell$ ranging from 1
  to $k$, such that $\psi$ is equivalent to $\bigvee_{\tau \in \theta}
  \phi_{\tau}$. Notice that if we would not require the types to be
  disjoint, we would only need to consider $k$-types, and not
  $\ell$-types, for all $\ell \leq k$. However, the latter
  restriction, and corresponding extension, will prove technically
  more convenient.

  Using the information that $A \models \psi$ is completely determined
  by the set of types $\theta_\psi$ realized in $A$, we now present
  our dynamic algorithm. It will maintain the following functions. For
  every disjoint $\ell$-type $\tau$, with $\ell \leq k$, and set $I =
  \{i_1,\ldots,i_{|I|}\} \subseteq \{1,\ldots,\ell\}$, let
$$f^I_\tau(x_1,\ldots,x_{|I|}) = |\{(a_1,\ldots,a_\ell) \mid
\type{A}{\bar{a}} = \tau \wedge \forall j \in [1,|I|]: a_{i_j} =
x_j\}|$$

Here we write $I = \{i_1,\ldots,i_{|I|}\}$ such that $i_j < i_{j+1}$,
for all $j \in [1,|I|-1]$. Then, for $I = \emptyset$,
$f^\emptyset_\tau$ defines the number of disjoint tuples in $A$ which
have type $\tau$. When $I = \{i_1, \ldots, i_{|I|}\} \neq \emptyset$,
and given $\bar{x} = (x_1,\ldots,x_{|I|})$, $f^I_\tau(\bar{x})$
defines the number of tuples in $A$ which (1) have type $\tau$ and (2)
have at positions $i_j$ exactly element $x_j$, for all $j \in
[1,|I|]$.

Notice that the numbers defined by the above functions can become
bigger than $n$, the number of universe elements, but are always
smaller than $n^k$. Hence, every such number can be encoded as a
number with $k$ digits in base $n$, which is exactly how our functions
will encode these numbers. Thereto, for every function $f^I_\tau$
mentioned above, there actually are $k$ functions $f^{I,1}_\tau, \ldots,
f^{I,k}_\tau$ each defining one digit of the desired number defined by
$f^I_\tau$. For clarity, we use the functions $f^I_\tau$ instead of
the actual ones encoding their digits.

As we are in the setting where precomputations is allowed, we can
assume that the functions are properly initialized. For any $l \in
[1,k]$, let $\tau_\neg$ be
the unique $l$-type containing only negated atoms, i.e. atoms of the form
$\neg R(\bar i)$. Then, for all $l$-types $\tau \neq \tau_\neg$, set
$I$, and tuple $\bar x$, initially
$$f^I_\tau(\bar x) = 0$$
and for $\bar x = (x_1,\ldots,x_{|I|})$ it holds that
$$f^I_{\tau_\neg}(\bar x) = 0 \mbox{ if $x_i = x_j$, for some $i \neq j$}$$
and
$$f^I_{\tau_\neg}(\bar x) = (n - |I|)\cdot(n - (|I|+1))\cdot \cdots \cdot
(n - l) \mbox{, otherwise}$$

We will now show how to incrementally maintain these
functions. Thereto, we first give the precomputed functions and
relations which will be used for the updates. For simplicity, we
assume the universe of size $n$ consists of the elements
$\{0,\ldots,n-1\}$. Then, there is a constant (0-ary function) $\minM$
denoting $0$, functions $\plusfun$
and $\minfun$ such that $\plusfun(x,y) = x+y \mbox{ }(\text{mod } n)$ and
$\minfun(x,y) = x-y \mbox{ }(\text{mod } n)$, and accompanying relations
$R_\plusfun$ and $R_\minfun$ such that $R_\plusfun(x,y)$ holds iff $x+y \geq
n$, and $R_\minfun(x,y)$ holds iff $x-y < 0$. That is, the functions
$\plusfun$ and $\minfun$ are defined on all parameters and count
modulo $n$. The accompanying relations $R_\plusfun$ and
$R_\minfun$ contain the additional information saying whether the
addition or subtraction indeed went above $n-1$ or below $0$. These
functions allow to define addition and subtraction on the $k$-digit
base-$n$ numbers used in the functions. Therefore, we simply perform
addition and subtraction on these numbers in the sequel.

Second, we introduce some additional notation. As before, we write
$\bar{x}$ for a tuple of elements, but abuse notation and also
denote the set of elements in $\bar{x}$ by $\bar{x}$, and,
correspondingly, apply set-theoretic operations on them, e.g. $\bar{x}
\cup \bar{y}$.

Further, for an integer $\ell$, set $I \subseteq \{1, \ldots, \ell\}$,
and tuples $\bar{x} = (x_1,\ldots,x_{|I|})$ and $\bar{y}$, we let an
\emph{indexing} for $\ell,I,\bar{x},\bar{y}$ be a function $\ind:
\bar{x} \cup \bar{y} \rightarrow \{1,\ldots,\ell\}$ such that for all
$j \in [1,|I|]$, $\ind(x_{j}) = i_j$. The indexing $\ind$ is
\emph{proper} iff for all $z,z' \in \bar{x} \cup \bar{y}$, $\ind(z) =
\ind(z') \iff z = z'$. Notice that while the fact whether a function
$\ind$ is an indexing only depends on $I$ and $\ell$, the fact whether
it is proper depends on the actual values of $\bar{x}$ and
$\bar{y}$. However, this can easily be tested by the following
formula:
\[
\phi_\ind(\bar{y},\bar{x}) = \bigwedge_{\substack{z,z' \in \bar{x} \cup
    \bar{y}\\\ind(z) = \ind(z')}} z = z' \wedge \bigwedge_{\substack{z,z' \in \bar{x} \cup
    \bar{y}\\\ind(z) \neq \ind(z')}} z \neq z'
\]

Given $\bar{x}$ and $\bar{y}$ and a proper indexing $\ind$, we write
$(\bar{x},\bar{y})_\ind$ for the sequence $(u_1,\ldots,u_m)$, for some
$m$, such that (1) $\bar{u}$ contains every element in $\bar{x} \cup
\bar{y}$ exactly once and (2) $\ind(u_i) < \ind(u_{i+1})$, for all $i
\in [1,m-1]$. Hence, $\bar{u}$ is obtained from $\bar{x} \cup \bar{y}$
by eliminating elements which are equal (and thus have the same
index), and ordering the elements by their index. Further, we write
$\ind(\bar{y})$ to denote the tuple
$(\ind(y_1),\ldots,\ind(y_m))$. Finally, for a type $\tau$, and
$R(\bar{i}) \notin \tau$, let $\tau + R(\bar{i})$ denote the type
obtained from $\tau$ by removing $\neg R(\bar{i})$ and adding
$R(\bar{i})$. When, $\neg R(\bar{i}) \notin \tau$, $\tau + \neg
R(\bar{i})$ is defined similarly by removing $R(\bar{i})$ and adding
$\neg R(\bar{i})$.

We are now ready to give the update functions. For clarity we write
the $\ite(\phi,t_1,t_2)$ construct as ``if $\phi$ then $t_1$ else
$t_2$''. Then, for relation symbol $R$, $\ell$-type $\tau$, with $\ell \leq
k$, and $I \subseteq \{1,\ldots,\ell\}$, let
\begin{multline*}
  \phi_{\ins_R}^{f^I_\tau}(\bar{y},\bar{x}) \equiv f_\tau^I(\bar{x})
  \\
  + \sum_{\substack{\ind \mbox{ for }
      \ell,I,\bar{x},\bar{y}\\R(ind(\bar{y})) \in \tau}} \mbox{ if }
  \phi_\ind(\bar{y};\bar{x}) \mbox{ then } f^{I \cup
    \ind(\bar{y})}_{\tau
    + {\neg R(ind(\bar{y}))}}(\bar{x},\bar{y})_\ind \mbox{ else } 0 \\
  - \sum_{\substack{\ind \mbox{ for } \ell,I,\bar{x},\bar{y}\\\neg
      R(ind(\bar{y})) \in \tau}} \mbox{ if } \phi_\ind(\bar{y};\bar{x})
  \mbox{ then } f^{I \cup \ind(\bar{y})}_{\tau}(\bar{x},\bar{y})_\ind
  \mbox{ else } 0
\end{multline*}
and, similarly, 
\begin{multline*}
  \phi_{\del_R}^{f^I_\tau}(\bar{y},\bar{x}) \equiv f_\tau^I(\bar{x})
  \\
  + \sum_{\substack{\ind \mbox{ for }
      \ell,I,\bar{x},\bar{y}\\\neg R(ind(\bar{y})) \in \tau}} \mbox{ if }
  \phi_\ind(\bar{y};\bar{x}) \mbox{ then } f^{I \cup
    \ind(\bar{y})}_{\tau
    + {R(ind(\bar{y}))}}(\bar{x},\bar{y})_\ind \mbox{ else } 0 \\
  - \sum_{\substack{\ind \mbox{ for } \ell,I,\bar{x},\bar{y}\\
      R(ind(\bar{y})) \in \tau}} \mbox{ if } \phi_\ind(\bar{y};\bar{x})
  \mbox{ then } f^{I \cup \ind(\bar{y})}_{\tau}(\bar{x},\bar{y})_\ind
  \mbox{ else } 0.
\end{multline*}

Intuitively, both formulas compute the number of tuples with the given
type $\tau$ in the same manner: take the number of tuples which used
to have type $\tau$, add those which obtained type $\tau$, and remove the
ones which had type $\tau$, but do not anymore.

We briefly explain the correctness of these formulas by arguing that
after an update $\ins_R(\bar y)$ for a tuple $\bar x$ the number of
tuples which did not have type $\tau$ but do after the update is
indeed equal to the number computed on the second line of the update
formula $\phi_{\ins_R}^{f^I_\tau}(\bar{y},\bar{x})$.

Thereto, let $\bar a = (a_1,\ldots,a_l)$ be a disjoint tuple
consistent with $\bar x$ and $I$, i.e. for all $j \in [1,|I|]$, $x_j =
a_{i_j}$. We denote the structure obtained from $A$ after the update
$\ins_R(\bar y)$ by $A'$. Now, suppose $\type{A}{\bar a} \neq \tau$,
but $\type{A'}{\bar a} = \tau$. This can only hold if $\bar y
\subseteq \bar a$ and thereby the insertion of $R(\bar y)$ has changed
the type of $\bar a$ in $A$. More precisely, if we define $\bar k =
k_1,\ldots,k_m$ such that for all $j \in [1,m]$, $y_j = a_{k_j}$, then
$\type{A}{\bar a} = \tau + \neg R(\bar k)$ must hold. Notice also that
$\bar k$ is uniquely defined due to the fact that $\bar a$ is
disjoint. Now $\bar k$, in turn, defines a proper indexing $\ind$ on
$\bar x$ and $\bar y$ as follows: for all $j \in [1,|I|]$, $\ind(x_j) = i_j$
(by definition) and for all $j \in [1,m]$, $\ind(y_j) = k_j$. In this
manner we can thus associate a unique proper indexing to all tuples $\bar a$
which did not have type $\tau$, but do now. Then, for any indexing $\ind$, the expression $f^{I \cup
  \ind(\bar{y})}_{\tau + {\neg
    R(ind(\bar{y}))}}(\bar{x},\bar{y})_\ind$ defines exactly all such
tuples with which $\ind$ is associated. By iterating over all proper
indexings we hence count exactly all desired tuples.

Finally, for the acceptance relation we simply have to check whether
there is a tuple in the new structure which has a type contained in
$\theta_\psi$: 
\[
\phi_{\ins_R}^{\accept}(\bar{y}) \equiv \bigvee_{\tau \in \theta_{\phi}}
\phi_{\ins_R}^{f^\emptyset_\tau}(\bar{y}) \neq 0 \qquad\text{ and }\qquad
\phi_{\del_R}^{\accept}(\bar{y}) \equiv \bigvee_{\tau \in \theta_{\phi}}
\phi_{\del_R}^{f^\emptyset_\tau}(\bar{y}) \neq 0.
\] 
\qed

\section{Conclusion}\label{sec:conclusion}

We have studied the dynamic complexity of formal languages and, by
characterizing the languages maintainable in \dynprop as exactly the regular
languages, obtained the first lower bounds for \dynprop. This yields a
separation of \dynprop from \dynQF and \dynFO. We proved that every context-free language can be maintained in \dynFO and investigated the power of functions for dynamic programs in maintaining specific context-free and non context-free languages. 

As a modest extension we also proved a lower bound for \dynprop with built-in successor
functions. Hence, we are now one step closer to proving lower bounds
for \dynFO, but, of course, a number of questions arise:
\begin{itemize}
\item Can the results on the Dyck languages be extended to show that
  an entire subclass of the context-free languages, such as the
  deterministic or unambiguous context-free languages, can be
  maintained in \dynQF?
\item We have seen that $D_1 \in \dynprop(\SetSucc,\Rel)$. Can it be
  shown that $D_2 \notin \dynprop(\SetSucc,\Rel)$?
\item Can some of the lower bound techniques for \dynprop be extended to
  \dynqf, in order to separate \dynqf from \dynFO, or at least from
  \dynp? Is there a context-free language that is not maintainable in \dynqf?
\end{itemize}

{\small
\bibliographystyle{abbrv}

}

\end{document}